\pdfoutput=1
\documentclass[a4paper,USenglish,cleveref, autoref, thm-restate, numberwithinsect]{lipics-v2021}

\usepackage{common}

\title{Abstract Congruence Criteria for Weak Bisimilarity}

\author{Stelios Tsampas}{KU Leuven, Belgium}{stelios.tsampas@cs.kuleuven.be}{https://orcid.org/0000-0001-8981-2328}{}

\author{Christian Williams}{University of California, Riverside, United States}{williams@math.ucr.edu}{https://orcid.org/0000-0001-5466-4708}{}

\author{Andreas Nuyts}{Vrije Universiteit Brussel, Belgium}{andreas.nuyts@vub.ac.be}{https://orcid.org/0000-0002-1571-5063}{}

\author{Dominique Devriese}{Vrije Universiteit Brussel, Belgium}{dominique.devriese@vub.be}{https://orcid.org/0000-0002-3862-6856}{}

\author{Frank Piessens}{KU Leuven, Belgium}{frank.piessens@cs.kuleuven.be}{https://orcid.org/0000-0001-5438-153X}{}

\authorrunning{S. Tsampas, C. Williams, A. Nuyts, D. Devriese, and F. Piessens} 

\Copyright{Stelios Tsampas, Christian Williams, Andreas Nuyts, Dominique
  Devriese, and Frank Piessens} 

\ccsdesc{Theory of computation~Categorical semantics}

\keywords{Structural Operational Semantics, distributive laws, weak bisimilarity} 

\acknowledgements{This research was partially funded by the Research Fund KU
  Leuven. Andreas Nuyts was supported by a grant of the Research Foundation – Flanders (FWO).}

\nolinenumbers 

\hideLIPIcs  

\EventEditors{Filippo Bonchi and Simon J. Puglisi}
\EventNoEds{2}
\EventLongTitle{46th International Symposium on Mathematical Foundations of Computer Science (MFCS 2021)}
\EventShortTitle{MFCS 2021}
\EventAcronym{MFCS}
\EventYear{2021}
\EventDate{August 23--27, 2021}
\EventLocation{Tallinn, Estonia}
\EventLogo{}
\SeriesVolume{202}
\ArticleNo{7}

\begin{document}

\maketitle

\begin{abstract}
  We introduce three general compositionality criteria over operational semantics
  and prove that, when all three are satisfied together, they guarantee weak bisimulation
  being a congruence. Our work is founded upon Turi and Plotkin's mathematical
  operational semantics and the coalgebraic approach to weak bisimulation by
  Brengos. We demonstrate each criterion with various examples of success and
  failure and establish a formal connection with the simply WB cool rule format
  of Bloom and van Glabbeek. In addition, we show that the three criteria induce
  lax models in the sense of Bonchi et al.
\end{abstract}

\section{Introduction}
\label{sec:intro}

The problem of \emph{full abstraction} for programming language
semantics~\cite{10.5555/218623.218633, DBLP:journals/tcs/Gordon99}, i.e.
the perfect agreement between a denotational and an operational specification,
has been both significant and enduring. It requires that the denotational
semantics, which bestows each program with a denotation, a meaning, is
sufficiently coarse that it does not distinguish terms behaving the same
operationally. At the same time, the denotational semantics must remain a
\emph{congruence}~\cite{DBLP:conf/amast/Glabbeek93}, to make the semantics
\emph{compositional}: the denotation of a composite term is fully determined by the
denotations of its subterms irrespective of their internal structure.

From an operational point of view, the choice of behavioral equivalence is
generally open. Bisimilarity, trace equivalence, weak bisimilarity etc. are all
potentially applicable. However, bisimilarity is often too strong for practical
purposes. For instance, labelled transition
systems~\cite{Winskel_Nielsen_1993} (LTS) may perform invisible steps which are
not ignored by bisimilarity as opposed to \emph{weak}
bisimilarity~\cite{DBLP:books/sp/Milner80, DBLP:books/daglib/0067019}. This is
taken a step further in programming language semantics, where a natural definition of
program equivalence is that of the \emph{largest adequate} (w.r.t. observing termination)
\emph{congruence} relation~\cite{10.5555/1076265}. This relation is also known as
\emph{contextual equivalence}, the crown jewel of program equivalences, and can
be reformulated in a more explicit manner as \emph{Morris-style} contextual
equivalence~\cite{morris, DBLP:journals/tcs/Gordon99,
  DBLP:journals/mscs/Pitts00}, i.e. indistinguishability under all program
contexts.

Despite being the subject of vigorous research, proving that coarser behavioral
equivalences are congruences remains a hard problem.
Weak bisimilarity in particular stands out as a popular equivalence that has
seen widespread usage in the literature yet has been proven hard to reason
with~\cite{DBLP:journals/ita/Rutten99, DBLP:conf/concur/SangiorgiM92,
  DBLP:conf/cav/BaierH97}. To that end, various powerful methods exist for
proving congruence-closedness of bisimilarity, like Howe's
method~\cite{DBLP:conf/lics/Howe89, DBLP:journals/iandc/Howe96} and \emph{logical
  relations}~\cite{DBLP:journals/corr/abs-1103-0510, DBLP:conf/lics/Pitts96,
  DBLP:journals/igpl/Pitts97}, yet the machinery involved is complicated and
non-trivial to execute correctly. In addition, there have
been the so-called \emph{cool congruence formats} for weak bisimulation
introduced by Bloom~\cite{DBLP:journals/tcs/Bloom95} and van
Glabbeek~\cite{DBLP:journals/tcs/Glabbeek11} that ensure weak bisimilarity
being a congruence, but only if the semantics adhere to the rule formats.

In this work, we propose three general \emph{compositionality criteria} over
operational specifications that, when all three are satisfied, guarantee weak
bisimilarity being a congruence.
The foundation of our approach is categorical: on the one hand, the
framework of \emph{mathematical operational semantics} introduced by Turi and
Plotkin~\cite{DBLP:conf/lics/TuriP97} acts as an ideal abstract setting to
explore programming language semantics. On the other hand, the coalgebraic
approach introduced by Brengos~\cite{DBLP:journals/corr/Brengos13,
  DBLP:journals/jlp/BrengosMP15} describes a seamless way to define weak
bisimulation in categories of coalgebras, should the underlying behavior be
monadic and equipped with an order structure. With that in mind, the three
criteria, which we name \emph{continuity}, \emph{unitality} and
\emph{observability}, essentially characterize how the semantics interact with
the order structure of the behavior.

\paragraph*{Related work}

As mentioned earlier, Bloom and van Glabbeek have introduced the cool congruence
formats; weak bisimulation for any system given in these formats is guaranteed to be
a congruence. In this work we are also able to establish a formal
connection with the \emph{simply WB cool} format for LTSs (\cref{th:simply}), specifically
that any system given in the simply WB cool rule format automatically
satisfies the three criteria. In that sense, our three criteria form a broader,
less restrictive approach on the
same problem (extended beyond LTSs). Furthermore, the many examples provided by van
Glabbeek~\cite{DBLP:journals/tcs/Glabbeek11} help explain this connection and at
the same time act as excellent hands-on case studies for our three criteria.

Apart from the aforementioned work of Brengos, various approaches at the hard
problem of coalgebraic weak bisimulation have been
proposed~\cite{DBLP:conf/calco/Popescu09, DBLP:journals/entcs/RotheM02,
  DBLP:conf/icalp/GoncharovP14, DBLP:conf/csl/BonchiPPR14}. Of particular
interest is the work of Bonchi et al.~\cite{DBLP:conf/concur/BonchiPPR15} on up-to
techniques~\cite{DBLP:journals/acta/BonchiPPR17} for weak bisimulations in the
context of mathematical operational semantics. The main theoretical device in
their work is that of a \emph{lax model}, a relaxation of the notion of a
\emph{bialgebra}. We relate lax models with our work by showing that
specifications satisfying the three criteria induce lax models
(\cref{laxtheorem}) and argue that, as a formal method, our criteria are
significantly easier to prove.

\paragraph*{Paper outline}

In \cref{sec:mathop} we introduce Turi and Plotkin's mathematical
operational semantics~\cite{DBLP:conf/lics/TuriP97} as well as the two ``running''
example systems used throughout the paper. In \Cref{sec:ord} we present the
work of Brengos on weak bisimulation~\cite{DBLP:journals/corr/Brengos13,
  DBLP:journals/jlp/BrengosMP15} and show how it applies to our two examples. We
expand on our contribution in \Cref{sec:ctxsem}, where we introduce the three
criteria and put them to the test against our main examples as
well as examples from van Glabbeek~\cite{DBLP:journals/tcs/Glabbeek11}. We
subsequently formalize the connection of our three criteria with the simply WB
cool rule format (\cref{th:simply}) and then move on to present our main theorem
(\Cref{the:main}). Finally, in \cref{subsec:lax}
we show how our three criteria induce lax models in the sense
of Bonchi et al.~\cite{DBLP:conf/concur/BonchiPPR15}.

\section{Preliminaries}

\subsection{Mathematical Operational Semantics}
\label{sec:mathop}

We first summarize the basic framework of Turi and Plotkin's \emph{mathematical
  operational semantics}~\cite{DBLP:conf/lics/TuriP97}. The idea is that operational
semantics correspond to~\emph{distributive laws} of varying complexity on a base
category $\mathbb{C}$. For our work we need only consider the most
important form of distributive laws, that of \emph{GSOS
  laws}~\cite{DBLP:conf/lics/TuriP97}, as they are in an 1-1 
correspondence with the historically significant \emph{GSOS} rule
format~\cite{DBLP:journals/jlp/Plotkin04a}.

\begin{definition}
  \label{def:gsos}
  Let endofunctors $\Sigma, B : \mathbb{C} \to \mathbb{C}$ on a cartesian
  category $\mathbb{C}$. A GSOS law of $\Sigma$ over  $B$ is a natural
  transformation $\rho : \Sigma (\Id \product B) \nat B\Sigma^{*}$, where
  $\Sigma^{*}$ is the free monad over $\Sigma$.
\end{definition}

Endofunctors $\Sigma$ and $B$ are understood as the syntax and behavior (resp.) of a
system whereas $\rho$ represents the semantics. Over the course of the paper we
will also be using an alternative representation of GSOS laws given by the
following correspondence.

\begin{proposition}
  \label{prop:gsos}
  GSOS laws $\rho : \Sigma (\Id \product B) \nat B\Sigma^{*}$ are
  \emph{equivalent} to natural transformations $\lambda : \Sigma^{*} (\Id
  \product B) \nat B\Sigma^{*}$ respecting the structure of $\Sigma^{*}$ as follows:
    \begin{center}
      \begin{tikzcd}
        \Sigma^{*}\Sigma^{*} (\Id \times B) \arrow[rr, "\Sigma^{*} \langle
        \Sigma^{*}\pi_{1} {,} \lambda \rangle",
        Rightarrow] \arrow[d, "\mu_{(\Id \times B)}", Rightarrow]
        & & \Sigma^{*} (\Id \times B) \Sigma^{*} \arrow[r, "\lambda_{\Sigma^{*}}", Rightarrow]
        & B \Sigma^{*}\Sigma^{*} \arrow[d, "B \mu", Rightarrow]
        \\
        \Sigma^{*} (\Id \times B) \arrow[rrr, "\lambda", Rightarrow]
        &
        & & B\Sigma^{*}
      \end{tikzcd}
      \hfil
      \begin{tikzcd}
        \Id \times B \arrow[d, "\eta_{(\Id \times B)}"', Rightarrow] \arrow[dr,
        "B\eta \circ \pi_{2}", Rightarrow]
        \\
        \Sigma^{*} (\Id \times B) \arrow[r, "\lambda", Rightarrow]
        & B \Sigma^{*}
      \end{tikzcd}
    \end{center}
  \end{proposition}

\begin{remark}
  We shall be calling both GSOS laws $\rho$ and natural transformations
  $\lambda$, as used in \cref{prop:gsos}, \emph{GSOS laws} for the sake of
  brevity.
\end{remark}

GSOS laws induce \emph{bialgebras}, i.e. algebra-coalgebra pairs that agree with
the semantics.

\begin{definition}
  \label{def:bialg}
  A \emph{bialgebra} for a GSOS law $\lambda : \Sigma^{*} (\Id \product B)
  \nat B\Sigma^{*}$ (resp. $\rho : \Sigma (\Id \product B)
  \nat B\Sigma^{*}$) is a  $\Sigma$-algebra, $B$-coalgebra pair $\bialg{\Sigma
    X}{g}{X}{h}{BX}$ that commutes with $\lambda$ ($\rho$):
  \begin{center}
    \begin{tikzcd}
      \Sigma^{*} X \arrow[r, "g^{\#}"] \arrow[d, "{\Sigma^{*} \langle 1,h \rangle}"]
      & X \arrow[r, "h"]
      & B X
      \\
      \Sigma^{*} (X \product BX) \arrow[rr, "\lambda"]
      & & B \Sigma^{*} X \arrow[u, "B g^{\#}"]
    \end{tikzcd}
    \begin{tikzcd}
      \Sigma X \arrow[r, "g"] \arrow[d, "{\Sigma \langle 1,h \rangle}"]
      & X \arrow[r, "h"]
      & B X
      \\
      \Sigma (X \product BX) \arrow[rr, "\rho"]
      & & B \Sigma^{*} X \arrow[u, "B g^{\#}"]
    \end{tikzcd}
  \end{center}
  Where $g^{\#}$ is the $\mathcal{EM}$-algebra induced by $g$. A bialgebra
  \emph{morphism} from $\bialg{\Sigma
    X}{g}{X}{h}{BX}$ to $\bialg{\Sigma Y}{j}{Y}{k}{BY}$ is a map $f : X \to Y$
  that is both a $\Sigma$-algebra and a $B$-coalgebra morphism.
\end{definition}

If $a : \luarrow{\Sigma A}{A}{\cong}{}$ is the initial $\Sigma$-algebra, the
algebra of terms, and $z :
\luarrow{Z}{BZ}{\cong}{}$ is the final $B$-coalgebra, the coalgebra of
behaviors, the following proposition promotes a GSOS law to an \emph{operational
  semantics} of a language, in the form of a morphism $f$ mapping programs
living in $A$ to behaviors in $Z$.

\begin{proposition}[From \cite{DBLP:journals/tcs/Klin11,DBLP:conf/lics/TuriP97}]
  \label{distr}
  Let endofunctors $\Sigma, B : \mathbb{C} \to \mathbb{C}$ on a cartesian
  category $\mathbb{C}$ such that $a : \luarrow{\Sigma A}{A}{\cong}{}$ is the
  initial $\Sigma$-algebra and $z : \luarrow{Z}{BZ}{\cong}{}$ is the final
  $B$-coalgebra. Every GSOS law $\lambda : \Sigma^{*}(\Id \product B) \nat
  B\Sigma^{*}$ induces a
  unique initial $\lambda$-bialgebra $\bialg{\Sigma A}{a}{A}{h}{BA}$, the
  \emph{operational model} of the language,
  and a unique final $\lambda$-bialgebra $\bialg{\Sigma Z}{g}{Z}{z}{BZ}$, the
  \emph{denotational model}. In addition, there exists a unique $\lambda$-bialgebra
  morphism $f : A \to Z$, mapping every program in $A$ to its behavior in $Z$.
\end{proposition}

The fact that map $f : A \to Z$ is an algebra homomorphism is a fundamental well-behavedness
property of GSOS laws, as it implies that
bisimilarity, defined as equality under $f$, is a \emph{congruence}, i.e. is
respected by all syntactic operators.

The categorical interpretation of GSOS rules can be better understood by the
following examples, which serve as the ``running'' examples of this paper.

\begin{example}[A \while~language]
  \label{ex:while}
  We introduce \while, a basic imperative language with a mutable state, whose
  syntax is generated by the following grammar:
\begin{grammar}\centering
  <prog> ::= \texttt{skip} | $v$ \texttt{:=} <expr> | <prog> ; <prog> |
  \texttt{while} <expr> <prog>
\end{grammar}

The statements are the standard \texttt{skip}, assignment \texttt{:=},
sequential composition \texttt{;} and \texttt{while}-loops. Expressions and
assignments act on a variable store whose type we denote as $S$ and we shall
be writing $[e]_{s}$ to denote evaluation of an expression $e$ under variable store
$s$.

The syntax of \while~corresponds to the $\Set$-endofunctor $\Sigma \triangleq
\term \uplus (V \times Exp) \uplus \Id^{2}
\uplus (Exp \times \Id)$, with the set of \while-programs $A$ as the
carrier of the initial $\Sigma$-algebra $a : \luarrow{\Sigma A}{A}{\cong}{}$.
The typical behavior functor  for deterministic systems with mutable state is $[S \times
(\{\checkmark\} \uplus \Id)]^{S}$, where $\{\checkmark\} \cong \term$ is
populated by the element that denotes \emph{termination}, but we will instead
use $T \triangleq [\mathcal{P}_{c}(S \times (\{\checkmark\} \uplus \Id))]^{S}$,
where $\mathcal{P}_{c}$ is the \emph{countable} power-set monad. This way the
behavior can be equipped with both a monadic and an order structure, given by
inclusion (see \cref{sec:ord}), while also allowing for non-determinism. The
carrier of the final coalgebra $z :
\luarrow{Z}{TZ}{\cong}{}$ is the set of behaviors acting on variable stores $S$
returning countably many new stores and possibly new behaviors.
\begin{gather*}
  \inference[skip]{}{\rets{s, \texttt{skip}}{s}} \qquad
  \inference[asn]{} {\rets{s, v~\texttt{:=}~e}{s_{[v \leftarrow [e]_{s}]}}} \qquad
  \inference[seq1]{\rets{s, p}{s\pr}}{\goes{s, p \texttt{;} q}{s\pr , q}} \\
  \inference[seq2]{\goes{s, p}{s\pr, p\pr}}{\goes{s, p \texttt{;} q}{s\pr , p\pr
      \texttt{;} q}} \qquad
  \inference[w1]{[e]_{s} = 0}{\rets{s, \texttt{while}~ e~p}{s}} \qquad
  \inference[w2]{[e]_{s} \neq 0 ~~ q \triangleq \texttt{while}~e~p}{\goes{s, q}{s ,
      p~;~q}}
\end{gather*}

The above semantics determines a GSOS law $\rho : \Sigma (\Id \times T) \nat
T\Sigma^{*}$ in the category $\Set$ of sets and (total) functions, or
equivalently a natural transformation $\lambda : \Sigma^{*} (\Id \times T) \nat
T\Sigma^{*}$. This example is covered in the
literature~\cite{DBLP:conf/ctcs/Turi97, DBLP:conf/cmcs/0001NDP20}, but we
include the definition for posterity.
\begin{definition}[GSOS law of \while]
\begin{align*}
  & \rho_{X} : \Sigma (X \times TX) & \to & \quad T\Sigma^{*}X  \\
  & (x,f)~\mathtt{;}~(y,g) & \mapsto & \quad \lambda s.(\{(s\pr , y)~|~(s\pr , \checkmark) \in f(s)\}~\cup~\{(s\pr, (x\pr~\mathtt{;}~y))~|~(s\pr , x\pr) \in f(s)\}) \\
  & \mathtt{while}~e~(x,f) & \mapsto & \quad \lambda~s.
                                       \begin{cases}
                                         \{(s, (x~\mathtt{;}~\mathtt{while}~e~x))\} \quad \mathrm{if}~[e]_{s} \neq 0\\
                                         \{(s, \checkmark)\} \quad \mathrm{if}~[e]_{s} = 0
                                       \end{cases} \\
  & \mathtt{skip} & \mapsto & \quad \lambda s. \{(s, \checkmark)\} \\
  & v~\mathtt{:=}~e & \mapsto & \quad \lambda s. \{(s_{[v \leftarrow ev]}, \checkmark)\}~\mathrm{for}~ev = [e]_{s}
\end{align*}
\end{definition}
To see how the semantics is connected to the GSOS law, consider rules seq1 and
seq2, corresponding to the first line in $\rho$. Roughly, writing $\goes{s,
  p}{s\pr , \checkmark}$ denotes that $(s\pr , \checkmark) \in f(s)$. The fact
that $\goes{s,p}{s\pr , \checkmark}$ and $\goes{s,p}{s\pr , q}$ are rule premises is
reflected in the construction of the set of transitions. Note also that
$q$ is used in the conclusion of rules seq1 and seq2. Similarly in $\rho$
we can see that $y$ is present in both the left and right side of $\rho$, which
is why the shape of the behavior of subterms is $X \times TX$ rather than simply $TX$.
\end{example}

\begin{example}[A simple process calculus]
  \label{ex:proc}
  We introduce a simple process calculus, or \spc~for short, based on the
  classic \emph{Calculus of communicating systems} introduced by
  Milner~\cite{DBLP:books/sp/Milner80}. Its syntax is generated by the following
  grammar:

  \begin{grammar}\centering
    <p> ::= \texttt{0} | $\delta$.<p> | <p> $\Vert$  <p> | <p> + <p>
  \end{grammar}

  From left to right we have the \emph{null process} \texttt{0}, \emph{prefixing} of a
  process $p$ with action $\delta \in  \Delta_{\tau} = \Delta \cup
  \overline{\Delta} \cup \{\tau\}$
  living in a set composed of actions $\Delta$, coactions $\overline{\Delta}$ and an
  \emph{internal} action $\tau$,
  \emph{parallel composition} and finally \emph{non-deterministic choice}. 
  \begin{gather*}
    \inference[sum1]{\goesv{P}{P\pr}{\delta}}{\goesv{P + Q}{P\pr}{\delta}} \qquad
    \inference[sum2]{\goesv{Q}{Q\pr}{\delta}}{\goesv{P + Q}{Q\pr}{\delta}} \qquad
    \inference[com1]{\goesv{P}{P\pr}{\delta}}{\goesv{P \Vert Q}{P\pr \Vert
    Q}{\delta}} \\
    \inference[com2]{\goesv{Q}{Q\pr}{\delta}}{\goesv{P \Vert Q}{P \Vert
        Q\pr}{\delta}} \quad
    \inference[syn]{\alpha \in \Delta & \goesv{P}{P\pr}{\alpha} & \goesv{Q}{Q\pr}{\overline{\alpha}}}{\goesv{P \Vert Q}{P\pr \Vert
        Q\pr}{\tau}} \quad
    \inference[prefix]{}{\goesv{\delta.P}{P}{\delta}}
  \end{gather*}
  Similarly to \cref{ex:while}, the semantics forms a GSOS law $\rho : \Sigma (\Id
  \times T) \nat T\Sigma^{*}$ in $\Set$ for $\Sigma \triangleq \top \uplus (\Delta_{\tau}
  \times \Id) \uplus \Id^{2} \uplus \Id^{2}$ and $T \triangleq
  \mathcal{P}_{c}(\Delta_{\tau} \times \Id)$. In this case, the carrier of the
  final coalgebra $z :\luarrow{Z}{TZ}{\cong}{}$
  is the set of all strongly extensional (meaning that distinct children of nodes
  are not bisimilar), countably-branching, $\Delta_{\tau}$-labelled
  trees~\cite[\S 4]{DBLP:journals/acs/AdamekLMMS15}.
  \begin{definition}[GSOS law of \spc]
    \begin{align*}
      & \rho_{X} : \Sigma (X \times TX) & \mapsto & \qquad T\Sigma^{*}X \\
      & \delta.(x , W) &\mapsto& \qquad \{(\delta, x)\} \\
      & (x,X)~\Vert~(y,Z) &\mapsto& \qquad
                                    \{(\delta,(x\pr~\Vert~y))~|~(\delta, x\pr) \in W\}
                                    ~\cup~\{(\delta,(x~\Vert~y\pr))~|~(\delta, y\pr) \in Z\} ~\cup \\
      & & & \qquad \{(\tau,(x\pr~\Vert~y\pr))~|~(\alpha, x\pr) \in W \wedge (\overline{a}, y\pr) \in Z \} \\
      & (x,W) + (y,Z) &\mapsto& \qquad W \cup Z
    \end{align*}
  \end{definition}
\end{example}


\subsection{Order-enrichment}
\label{sec:ord}

Order-enriched categories~\cite{DBLP:journals/tcs/Wand79} typically equip their
hom-sets with an \emph{order} structure. While there are many forms of
order-enrichment, one ``nice'' such form that is convenient for our purposes is
$\cpov$-enrichment.

\begin{definition}[{\cite[\S 2.3]{DBLP:journals/jlp/BrengosMP15}}]
  A category $\mathbb{C}$ is $\cpov$-enriched when
  \begin{itemize}
  \item Every hom-set $\mathbb{C}(X,Y)$ carries a partial order $\leq$ and has
    all finite joins $\vee$.
  \item Composition is left-distributive over finite joins, i.e. given any
    morphisms $f,g,i$ with suitable domains and codomains, $i \circ (f \vee g)
    = i \circ f \vee i \circ g$.    
  \item Every ascending $\omega$-chain $f_{0} \leq f_{1} \leq\ldots$ for $f_{i}
    \in \mathbb{C}(X,Y)$ has a supremum $\bigvee_{i} f_{i} \in \mathbb{C}(X,Y)$.
  \item Composition $-\circ- : \mathbb{C}(X,Y) \times \mathbb{C}(Y,Z) \to
    \mathbb{C}(X,Z)$ is continuous, meaning that for any ascending $\omega$-chains
    $f_{i}, g_{i}$ and morphisms $f,g$ with suitable (co)domains, we have
    $$g \circ \bigvee_{i} f_{i} = \bigvee_{i}(g \circ f_{i})~~\mathrm{and}~~(\bigvee_{i}
    g_{i}) \circ f = \bigvee_{i}(g_{i} \circ f)$$
  \end{itemize}
\end{definition}

For reasons that will become clear in \cref{subsec:sat}, we are interested in
monads whose Kleisli category is $\cpov$-enriched as we are 
looking to use them as behaviors in our distributive laws and exploit their
order structure. Examples of such monads are the powerset $\mathcal{P}$~\cite[\S
4.1]{DBLP:journals/jlp/BrengosMP15}, the countable powerset $\mathcal{P}_{c}$~\cite[\S
4.3]{DBLP:journals/jlp/BrengosMP15} and the convex combination monad $\mathcal{CM}$~\cite[\S
4.3]{DBLP:journals/jlp/BrengosMP15}.

\begin{example}[Continuation of \cref{ex:while}]
  \label{rem:monad}
  The monad structure $\langle T, \eta, \mu  \rangle$ for $T \triangleq
  [\mathcal{P}_{c}(S \times (\{\checkmark\} \uplus \Id))]^{S}$ is given by $
  \eta(x)(s) = \{(s, x)\}$ and $\mu(f)(s) = \bigcup_{(s\pr, g) \in f(s)} \begin{cases}
    (s\pr,\checkmark) & \mathrm{if}~g = \checkmark \\
    g(s\pr) & \mathrm{if}~g \neq \checkmark
  \end{cases}$.

  The $\cpov$-enrichment of $\kl{T}$ stems from the fact that the countable powerset monad
  $\mathcal{P}_{c}$ is itself~$\cpov$-enriched\footnote{We (slightly abusively) say that a monad $T$ is
    $\cpov$-enriched whenever $\kl{T}$ is.}. More specifically, we have $f \leq g \iff
  \forall x,s. f(x)(s) \subseteq g(x)(s)$, $f \vee g \triangleq \lambda
  x.\lambda s.[f(x)(s) \cup g(x)(s)]$ and $\bigvee_{i} f_{i} = \lambda x.\lambda
  s.\bigcup_{i} f_{i}(x)(s)$.
\end{example}

\begin{example}[Continuation of \cref{ex:proc}]
  \label{ex2}
  The monad structure $\langle T, \eta, \mu \rangle$ for $T \triangleq
  \mathcal{P}_{c}(\Delta_{\tau} \times \Id)$ was developed by Brengos~\cite[\S
  4.1]{DBLP:journals/corr/Brengos13} as a central,
  demonstrative example of the role of monads in coalgebraic weak bisimulation.
  The unit is simply $\eta(x) = \{(\tau, x)\}$ and the join $\mu_{X} :
  \mathcal{P}_{c}[\Delta_{\tau} \times \mathcal{P}_{c}(\Delta_{\tau} \times X)] \to
  \mathcal{P}_{c}(\Delta_{\tau} \times X)$ is $\mu = \mu_{\mathcal{P}_{c}} \circ \mathcal{P}_{c}\mu_{\Delta} \circ
  \mu_{\mathcal{P}_{c}} \circ \mathcal{P}_{c}st$ where $st_{X,Y} : X \times \mathcal{P}_{c}Y \to
  \mathcal{P}_{c} (X \times Y)$ is the \emph{tensorial strength} of $\mathcal{P}_{c}$
  given by
  $$st_{X,Y} (x, Y) = \{(x, y)~|~y \in Y \}$$ and $\mu_{\Delta} : \Delta_{\tau} \times \Delta_{\tau} \times X \to
  \mathcal{P}_{c} \Delta_{\tau} X$ is   $$\mu_{\Delta} =
                                            \begin{cases}
                                              (\delta , \tau , x)  \mapsto \{(\delta, x)\} \\
                                              (\tau , \delta , x)  \mapsto \{(\delta, x)\} \\
                                              (\delta_{1} , \delta_{2} , x)  \mapsto \emptyset \quad \text{when} \quad \delta_{1},\delta_{2} \neq \tau.
                                            \end{cases}$$

  In other words, $\mu$ will disallow any two-step transition that outputs two
  visible labels in a row and allow everything else, but only after redundant
  invisible labels are removed. The motivation behind the definition
  of $\mu$ will become clear in \cref{subsec:sat}. As for the $\cpov$-enrichment
  of $\kl{T}$, it is also a consequence of
  $\kl{\mathcal{P}_{c}}$ being $\cpov$-enriched, with $f \leq g \iff f(x) \subseteq
  g(x)$, $f \vee g \triangleq \lambda x.[f(x) \cup g(x)]$ and $\bigvee_{i} f_{i}
  = \lambda x.\bigcup_{i} f_{i}(x)$.
\end{example}

\subsection{Free monads in $\cpov$-enriched categories}
\label{subsec:sat}

We now turn our attention to the coalgebraic approach to weak bisimulation
of Brengos~\cite{DBLP:journals/corr/Brengos13, DBLP:journals/jlp/BrengosMP15}.
The main idea is that given an endomorphism $\alpha : X \to X$ in an
$\cpov$-enriched category, there exists the \emph{free monad} over $\alpha$,
$\alpha^{*}$ . This process, which we call the \emph{reflexive transitive
 closure} of $a$ or simply the \emph{rt-closure} of $a$, can be used to derive
saturated transition systems or the \emph{multi-step} evaluation relation of a
programming language.

\begin{remark}
  In general, a \emph{monad} in a bicategory $K$ is an endomorphism $\epsilon : X \to
  X$ equipped with 2-cells $\eta : \id_{X} \to \epsilon$  and $\mu : \epsilon
  \circ \epsilon \to \epsilon$ subject to the conditions $\mu \circ \epsilon \eta =
  \mu \circ \eta \epsilon = \id_{\epsilon}$ and $\mu \circ \epsilon \mu = \mu \circ
  \mu \epsilon $. One can recover
  the classic definition of a monad by taking the strict 2-category $\Cat$ of
  categories, functors and natural transformation as $K$. In order-enriched
  categories, where there is at
  most one 2-cell between morphisms, usually denoted by $\leq$, monads are
  simply endomorphisms $\epsilon : X \to X$ satisfying $\id_{X} \leq \epsilon$
  and $\epsilon \circ \epsilon \leq \epsilon$. Monads in order-enriched
  categories are known as \emph{closure operators}.
\end{remark}

\begin{definition}[{\cite[Definition 3.6]{DBLP:journals/jlp/BrengosMP15}}]
  \label{def:free}
  An order-enriched category $\mathbb{K}$ \emph{admits free monads} if for any
  endomorphism $\alpha : X \to X$ there exists a monad $\alpha^{*}$ such that
  \begin{itemize}
  \item $\alpha \leq \alpha^{*}$
  \item if $\alpha \leq \beta$ for a monad $\beta : X \to X$ then $\alpha^{*} \leq \beta$.
  \end{itemize}
\end{definition}

The free monad of an endomorphism in~$\cpov$-enriched categories is the least
solution of a certain assignment or, equivalently, the supremum of an ascending
$\omega$-chain.

\begin{proposition}[{\cite[\S 3.2]{DBLP:journals/jlp/BrengosMP15}}]
  \label{sat}
  For an $\cpov$-enriched category $\mathbb{C}$, the free monad $(-)^{*} :
  \forall X.   \mathbb{C}(X,X) \to \mathbb{C}(X,X)$ of $\alpha : X
  \to X$ is given by $\alpha^{*} \triangleq
  \mu x.\id \vee x \circ \alpha = \bigvee_{n < \omega}(\id \vee \alpha)^{n}$.
\end{proposition}

Free monads in~$\cpov$-enriched categories enjoy a number of interesting properties.

\begin{lemma}[Properties of free monads]
  \label{satprop}
  For all $\alpha : X \to X$, $\beta : Y \to Y$, we have
  \begin{enumerate}[i.]
  \item $\forall f : X \to Y.~f \circ \alpha \leq \beta \circ f \implies f
    \circ \alpha^{*} \leq \beta^{*} \circ f$
  \item $\alpha^{**} = \alpha^{*}$
  \item $\alpha^{*} = \alpha^{*} \circ \alpha^{*}$
  \item $\id^{*} = \id$.
  \end{enumerate}
\end{lemma}

We can now revisit \cref{ex:while} and \cref{ex:proc} to witness how the
rt-closure acts on each operational model $h : A \to TA$. Note that
throughout the paper we shall be using the notation $-\diamond-$ to denote
composition in Kleisli categories.

\begin{example}[Continuation of \cref{ex:while}]
  The rt-closure (free monad) $h^{*}$ of the operational model $h : A \to TA$
  amounts to the reflexive, transitive closure of $h$. The initial stage $(\id
  \vee h)^{0} = \id_{\kl{T}} = \eta_{A}$ takes care of the \emph{reflexive}
  step, while stages $(\id \vee h)^{n + 1} = (\id \vee h)^n \diamond (\id \vee
  h) = (\id \vee h)^{n} \vee ((\id \vee h)^{n} \diamond h)$ amount to the
  inductive, \emph{transitive} step, which acts according to the definition of
  monadic composition. For instance, (s, \texttt{skip ; skip}) weakly
  transitions to (s, \texttt{skip ; skip}), (s, \texttt{skip}) and (s ,
  \checkmark).
\end{example}

\begin{example}[Continuation of \cref{ex:proc}]
  The rt-closure $h^{*}$ amounts to the saturation of $h : A \to TA$. The unit $\eta$
  establishes $\Goesv{p}{p}{\tau}$ as a silent step in $h^{*}$ for $p \in A$
  whereas $\mu$ ensures that one-step transitions in the saturated system
  \emph{cannot} produce more than one visible label, i.e. they will always be of
  the sort of $\Goesv{p}{\goesv{q}{\Goesv{r}{s}{\tau^{*}}}{\delta}}{\tau^{*}}$
  or $\Goesv{p}{q}{\tau^{*}}$. The reader may refer
  to~\cite{DBLP:journals/corr/Brengos13, DBLP:journals/jlp/BrengosMP15} for more details.
\end{example}

\subsection{Weak bisimulation}

We are now ready to define \emph{weak bisimulation} as a special case of an
\emph{Aczel-Mendler bisimulation}~\cite{DBLP:conf/ctcs/AczelM89,
  DBLP:journals/corr/abs-1101-4223}.

\begin{definition}
  \label{def:am-bisim}
  Let $f : X \to FX$ be a coalgebra for a functor $F : \mathbb{C} \to
  \mathbb{C}$. An \emph{Aczel-Mendler bisimulation}, or simply
  \emph{bisimulation}, for $f$ is a relation (span) $X \xleftarrow{r_1} R
  \xrightarrow{r_2} X$ which is the carrier of a coalgebra $e : R \to FR$ that
  lifts to a span of coalgebra homomorphisms, i.e. making the following diagram
  commute.
	\begin{center}
		\begin{tikzcd}
      FX
			& FR
      \arrow{l}[swap]{Fr_1}
      \arrow{r}{Fr_2}
			& FX
			\\
      X \arrow{u}{f}
			& R
      \arrow{u}{e}
      \arrow{l}[swap]{r_1}
      \arrow{r}{r_2}
			& X
      \arrow{u}{f}
		\end{tikzcd}
	\end{center}
\end{definition}

In our setting we elect to use the ``unoptimized'' definition of weak
bisimulation as bisimulation on a saturated system, mainly due to its simplicity.
Regardless, under the mild condition that \emph{arbitrary cotupling in $\kl{T}$
  is monotonic}, which is true in all of our examples, this definition of weak
bisimulation coincides with the traditional one~\cite[\S 6]{DBLP:journals/corr/Brengos13}.

\begin{definition}
  \label{def:weak}
  Assume a coalgebra $h : X \to TX$ for an $\cpov$-enriched monad $\langle T,
  \eta, \mu \rangle$. A \emph{weak bisimulation} for $h$ is a bisimulation for $h^{*}$.
\end{definition}

Two elements $x,y \in X$ are \emph{(weakly) bisimilar}, written as ($x
\approx y$) $x \sim y$, if there is a (weak) bisimulation that contains them. If $z : Z \to
TZ$ is the final $T$-coalgebra, then for every coalgebra $h : X \to TX$ we write
$h^{!} : X \to Z$ for the unique coalgebra homomorphism from $h$ to $z$ and
$h^{\dag}$ for the one from $h^{*}$ to $z$. The following
theorem presents the principle of weak coinduction, i.e. that weakly bisimilar
elements are mapped to the same weak behavior.

\begin{theorem}[{\cite[Theorem 6.8]{DBLP:journals/corr/Brengos13}}]
  \label{th:weakcoind}
  Let $\langle T, \eta, \mu \rangle$ be an $\cpov$-enriched monad with final
  coalgebra $z : Z \to TZ$. If $T$ preserves weak
  pullbacks, then the greatest weak bisimulation for a coalgebra $h : X \to TX$
  exists and coincides with the pullback of the equality span $\langle
  \id_{Z},\id_{Z} \rangle : Z \rightarrowtail Z \times Z$ along $h^{\dag} \times
  h^{\dag}: X \times X \to Z \times Z$.
\end{theorem}

The $\dag$ construction can be applied to any coalgebra, including the final
coalgebra $z : Z \to TZ$. The following lemma shows that $h^{\dag}$
can also be recovered via $z^{\dag}$ and $h^{!}$, a fact that will turn out to be useful
in \cref{sec:ctxsem}.

\begin{lemma}[{\cite[Lemma 6.9]{DBLP:journals/corr/Brengos13}}]
  \label{lem:weak}
  For \emph{any} $h : X \to TX$, we have $h^{\dag} = z^{\dag} \circ h^{!}$.
\end{lemma}


\section{Weak bisimulation congruence semantics}
\label{sec:ctxsem}

The theory introduced in \cref{sec:ord} sets the stage for our three
compositionality criteria, which ensure that weak bisimilarity is a congruence.
In addition, we shall test the criteria on \while~and \spc~and establish a
correspondence with the simply WB cool rule format of van
Glabbeek (see \cite{DBLP:journals/tcs/Glabbeek11}, definition also
included in \cref{app:simply}),
both formally and through examples, as it also guarantees weak bisimilarity
being a congruence and comes with less overhead compared to the WB cool format.
Finally, we shall briefly touch on the work of
Bonchi et al.~\cite{DBLP:conf/concur/BonchiPPR15} and show that systems
satisfying the three criteria induce lax models.

\subsection{The three compositionality criteria}
\label{subsec:crits}

Simply put, the three criteria ensure that semantics interact with the
order structure of the behavior functor in a sensible way. They are
\emph{abstract} in that they apply to any GSOS law $\lambda : \Sigma^{*} (\Id
\times T) \nat T\Sigma^{*}$ when $\langle T, \eta, \mu \rangle$ is an
$\cpov$-enriched monad. As such, we assume the existence of such a $T$ and
$\lambda$ throughout \cref{sec:ctxsem}. We present and explore each criterion
individually starting with the simplest of the three.

\begin{criterion}[Continuity]
  \label{crit1}
  For any ascending $\omega$-chain $f_{0}
  \leq f_{1} \leq\ldots : X \to TX$ the following condition applies:
  \begin{equation}
    \label{eq:crit1}
    \lambda \circ \Sigma^{*} \langle \id, \bigvee_{i} f_{i}\rangle = \bigvee_{i}
    \lambda \circ \Sigma^{*} \langle \id,  f_{i}\rangle
  \end{equation}
  Alternatively, we can write the above using $\rho : \Sigma (\Id \times T) \nat
  T\Sigma^{*}$ as
  \begin{equation}
    \label{eq:crit1a}
    \rho \circ \Sigma \langle \id, \bigvee_{i} f_{i}\rangle = \bigvee_{i}
    \rho \circ \Sigma \langle \id,  f_{i}\rangle.
  \end{equation}
\end{criterion}

\begin{proposition}
  \label{prop:cont}
  \cref{eq:crit1} and \cref{eq:crit1a} are equivalent.
\end{proposition}

We can compare \cref{crit1} to the \emph{local continuity} property of
lifted functors in Kleisli
categories~\cite[\S2.3]{DBLP:journals/lmcs/HasuoJS07}, i.e. lifted functors
respecting the $\cpov$-enrichment structure of the category. When it comes to
semantics, the following example from van Glabbeek~\cite[Example
2]{DBLP:journals/tcs/Glabbeek11} underlines what sort of rules may
violate \cref{crit1}.

\begin{example}[Illegal rules]
  \label{ex:illegal}
  Let us extend \spc~with a new unary operator, $\textcolor{MidnightBlue}{\lfloor}
  \langle p \rangle \textcolor{MidnightBlue}{\rfloor}$,
  subject to the following GSOS rules applying to specific actions $\alpha, b \in \Delta$:
  \begin{center}
$    \inference[\textcolor{MidnightBlue}{pos}]{\goesv{P}{P\pr}{\tau}}{\goesv{\textcolor{MidnightBlue}{\lfloor} P \textcolor{MidnightBlue}{\rfloor}}{\textcolor{MidnightBlue}{\lfloor} P\pr \textcolor{MidnightBlue}{\rfloor}}{\tau}} \quad
      \inference[\textcolor{MidnightBlue}{neg}]{\ngoesv{P}{a}}{\goesv{\textcolor{MidnightBlue}{\lfloor}
          P \textcolor{MidnightBlue}{\rfloor}}{\mathtt{0}}{b}}$
  \end{center}

  Rule \textcolor{MidnightBlue}{neg} allows a process that cannot perform an
  $a$-transition to terminate with a $b$-transition. This rule
  violates \cref{crit1}, which can be witnessed by testing the criterion with
  $(f_{0}(p) = \{(\tau, p\pr)\}) \leq (f_{1}(p) = \{(a, q), (\tau, p\pr)\}),$ as
  only $f_{0}$ is able to induce a $b$-transition. In addition, weak bisimulation
  fails to be a congruence as $\tau.\alpha.\mathtt{0} \approx \alpha.\mathtt{0}$
  but $\textcolor{MidnightBlue}{\lfloor}
  \tau.\alpha.\mathtt{0} \textcolor{MidnightBlue}{\rfloor} \not\approx
  \textcolor{MidnightBlue}{\lfloor} \alpha.\mathtt{0}
  \textcolor{MidnightBlue}{\rfloor}$.
\end{example}

Rule \textcolor{MidnightBlue}{neg} in the above example includes a negative
premise. A GSOS specification that has no negative premises - conclusions are never
negative - is called \emph{positive}. Positive GSOS specifications correspond to
\emph{monotone} GSOS laws (see~\cite{fiorepositive} and also \cite[Equation
7]{DBLP:conf/concur/BonchiPPR15}), in the sense that $g \leq f \implies \lambda \circ \Sigma^{*}
\langle \id {,} g \rangle \leq \lambda \circ \Sigma^{*} \langle \id {,} f
\rangle$. This is weaker than continuity (\Cref{crit1}), in that continuous GSOS
laws are monotone but not the other way around, but one has to look
very hard to find semantics that are monotone yet not continuous.

\begin{example}[A non-continuous monotone rule]
  We substitute the set of visible actions of~\spc~with $\mathbb{N}$ and
  define a new operator, $\textcolor{Plum}{\lfloor} \_
  \textcolor{Plum}{\rfloor}$, subject to rule $    \inference[\textcolor{Plum}{mon}]{\goesv{P}{}{\infty}}{\goesv{\textcolor{Plum}{\lfloor} P \textcolor{Plum}{\rfloor}}{\mathtt{0}}{\tau}}$,
  where $\goesv{P}{}{\infty}$ denotes that $P$ can perform an
  infinite number of transitions, i.e. set $\{(\delta,
  P\pr)~|~\goesv{P}{P\pr}{\delta}\}$ is infinite. Even though this is a monotone
  rule, it is not
  continuous. Consider for instance the ascending $\omega$-chain $f_{i} : X \to TX$ for some
  set of processes $X$ with $f_{i}(x) = \{(0,x),(1,x)\ldots(i,x)\}$. Notice that
  $(\tau, x) \subseteq \lambda \circ \Sigma^{*} \langle \id, \bigvee_{i}
  f_{i}\rangle$, but $(\tau, x) \not\subseteq \bigvee_{i} \lambda \circ \Sigma^{*}
  \langle \id,  f_{i}\rangle$.
\end{example}

\begin{example}[Continuation of \cref{ex:while}]
  Back to our \while~language, \Cref{crit1} is trivially true for all
  terms except for sequential composition, which is slightly more involved.
  In this case, we can see that $\rho_{X}((x,\bigvee_{i}
  f_{i}(x))~\mathtt{;}~(y,\bigvee_{i} f_{i}(y)))$ is mapped to
  \begin{multline*}
    \lambda s.(\{(s\pr , y)~|~(s\pr , \checkmark) \in (\bigvee_{i}
    f_{i}(x))(s)\}~\cup~\{(s\pr, (x\pr~\mathtt{;}~y))~|~(s\pr , x\pr) \in
    (\bigvee_{i} f_{i}(y))(s)\}) \\
    = \lambda s.\bigvee_{i}(\{(s\pr , y)~|~(s\pr , \checkmark) \in
    f_{i}(x)(s)\}~\cup~\{(s\pr, (x\pr~\mathtt{;}~y))~|~(s\pr , x\pr) \in
    f_{i}(y)(s)\}),
  \end{multline*}
  which is precisely $\bigvee_{i}(\rho_{X}((x, f_{i}(x))~\mathtt{;}~(y, f_{i}(y))))$.
\end{example}

\begin{example}[Continuation of \cref{ex:proc}]
  The transitions of prefix expressions such as $\delta P$ are, for any given
  $P$, independent of (the transitions of) $P$ thus the prefix rule is trivially
  continuous. Transitions for parallel composition and non-deterministic choice
  are basically unions of transitions of their subterms and hence satisfy \cref{crit1}.
\end{example}

\begin{criterion}[Unitality]
  \label{crit2}
  For any $f : X \to TX$,
  \begin{center}
    \begin{tikzcd}
      \Sigma^{*} X
      \arrow[d, "\Sigma^{*} \langle \id {,} \eta_{X} \vee f \rangle"']
      \arrow[dr, bend left, "(\lambda_{X} \circ \Sigma^{*} \langle \id {,} f \rangle)^{*}"]
      \arrow[dr, phantom, "\leq"]
      \\
      \Sigma^{*} (X \product TX)
      \arrow[r, "\lambda_{X}"']
      & T\Sigma^{*} X
    \end{tikzcd}
  \end{center}
\end{criterion}

This criterion characterizes how the semantics deal with internal steps, here
represented by the monadic unit $\eta$. The right path on the diagram represents
the rt-closed (weak) transitions of a composite term, the subterms of which
(strongly) transition according to $f$. On the other hand, the left path
represents the strong transitions of a composite term, the subterms of which
may also perform internal steps. This criterion, which somewhat resembles
the identity condition for lifting functors to Kleisli categories~\cite[\S
2.2]{DBLP:journals/lmcs/HasuoJS07}, dictates that adding arbitrary internal
steps to subterms should not lead to extraneous, meaningful observations.

\begin{remark}
  \label{crit2str}
  A slightly stronger but simpler formulation of \Cref{crit2} based on $\rho$ is
  \begin{center}
    \begin{tikzcd}
      \Sigma X
      \arrow[d, "\Sigma \langle \id {,} \eta_{X} \vee f \rangle"']
      \arrow[dr, bend left, "(\rho_{X} \circ \Sigma \langle \id {,} f \rangle)
      \vee (\eta_{X} \circ \theta_{X})"]
      \arrow[dr, phantom, "\leq"]
      \\
      \Sigma (X \product TX)
      \arrow[r, "\rho_{X}"']
      & T\Sigma^{*} X
    \end{tikzcd}
  \end{center}
  where $\theta : \Sigma \nat \Sigma^{*}$ is the universal natural transformation
  sending $\Sigma$ to its free monad. Compared to the original criterion, which
  asks for the transitions induced by internal steps to \emph{eventually} appear on
  the right side, this version asks for said transition to appear either on
  \emph{step 0}, the \emph{reflexive}, identity step (hence the added $\eta_{X} \circ
  \theta_{X}$), or \emph{step 1}, i.e. the transitions induced immediately by $f$. We
  can apply similar logic to \cref{prop:cont} to show that this version of
  \cref{crit2} is stronger.
\end{remark}

\Cref{crit2} works in the same way as the ``patience rule'' requirement of the
simply WB cool rule format~\cite[Definition 8, item
2]{DBLP:journals/tcs/Glabbeek11}, which dictates that the only rules with
$\tau$-premises are patience rules. For instance, the patience rule for a unary
operator $o(p)$ is
$\inference{\goesv{p}{p\pr}{\tau}}{\goesv{o(p)}{o(p\pr)}{\tau}}$. It is clear
that patience rules are achieving the same effect of forcing composite terms
to \emph{only} relay silent steps of subterms.

\begin{example}[{\cite[Example 4]{DBLP:journals/tcs/Glabbeek11}}]
  \label{ex:glab4}
  We extend~\spc~with $\textcolor{DarkOrchid}{\lfloor} \_
  \textcolor{Plum}{\rfloor}$ and introduce the following impatient rules:
  \begin{center}
    $\inference[\textcolor{Plum}{pat}]{\goesv{P}{P\pr}{\tau}}{\goesv{\textcolor{Plum}{\lfloor} P \textcolor{Plum}{\rfloor}}{\textcolor{Plum}{\lfloor} P\pr \textcolor{Plum}{\rfloor}}{\tau}} \quad
      \inference[\textcolor{Plum}{imp}]{\goesv{P}{P\pr}{\tau}}{\goesv{\textcolor{Plum}{\lfloor} P \textcolor{Plum}{\rfloor}}{\textcolor{Plum}{\lfloor} P\pr \textcolor{Plum}{\rfloor}}{c}}$
  \end{center}
  Rule \textcolor{Plum}{imp} violates \cref{crit2}, as taking  $f = \lambda
  x.\varnothing$ will not induce the $c$-transition present in the left path.
  Weak bisimulation fails to be a congruence as $\texttt{0} \approx \tau.\texttt{0}$ but
  $\textcolor{Plum}{\lfloor} \texttt{0}
  \textcolor{Plum}{\rfloor} \not\approx
  \textcolor{Plum}{\lfloor} \tau.\texttt{0}
  \textcolor{Plum}{\rfloor}$.
\end{example}

\begin{example}[Continuation of \cref{ex:while}]
  Language~\while~actually respects the stronger version of \cref{crit2} found
  in \cref{crit2str}. For \texttt{skip}, assignment and \texttt{while}-loops the
  transitions induced by $\rho_{X} \circ \Sigma \langle \id {,} f \rangle$ are
  the same regardless of $f$. This is not the case for sequential composition,
  but we observe that the left path always leads to $\goes{s, x \texttt{;} y}{s , x
    \texttt{;} y}$, which is covered by the added $\eta_{X} \circ \theta_{X}$ on
  the right path.
\end{example}

\begin{example}[Continuation of \cref{ex:proc}]
  \label{lax:nonex}
  \spc~provides for a good example of failure of \cref{crit2} as it echoes the
  well-known fact that weak bisimilarity is not compatible with
  non-deterministic choice in a manner similar to \cref{ex:glab4}. In particular,
  the left path always assigns $x + y$ transitions $\{(\tau, x), (\tau, y)\}$ but for $f = \lambda x.\varnothing$
  we have $(\lambda_{X} \circ \Sigma^{*} \langle \id , f \rangle)^{*} = \{(\tau, x +
  y)\}$. Clearly $\{(\tau, x), (\tau, y)\} \nsubseteq \{(\tau, x + y)\}$ and
  so \cref{crit2} is not satisfied. We can witness the incompatibility of
  non-deterministic choice by taking a cue from the failing instance and use a
  process which has no transitions, i.e. the null process: $\texttt{0} \approx
  \tau.\texttt{0}$ but $\delta.\texttt{0} + \texttt{0} \not\approx
  \delta.\texttt{0} + \tau.\texttt{0}$.
\end{example}

\begin{criterion}[Observability]
  \label{crit3}
  For any $f : X \to TX$,
  \begin{equation}
    \label{observ1}
    \lambda_{X} \circ \Sigma^{*} \langle \id {,} f \diamond f \rangle \leq (\lambda_{X} \circ \Sigma^{*} \langle \id {,} f \rangle)^{*}
  \end{equation}
  Equivalently, we can reformulate the above as
  \begin{equation}
    \label{observ2}
    \rho_{X} \circ \Sigma \langle \id {,} f \diamond f \rangle \leq (\lambda_{X} \circ \Sigma^{*} \langle \id {,} f \rangle)^{*} \circ \theta_{X}
  \end{equation}
  where $\theta : \Sigma \nat \Sigma^{*}$ is the universal natural transformation
  sending $\Sigma$ to its free monad.
\end{criterion}

\begin{remark}
  The fact that \cref{observ1} and \cref{observ2} are equivalent can be proved
  in a manner similar to \cref{prop:cont}.
\end{remark}

This criterion is roughly a weakening of the associativity condition for
liftings to Kleisli categories~\cite[\S 2.2]{DBLP:journals/lmcs/HasuoJS07}. We
can think of $f \diamond f$ as a two-step transition applied to subterms and
$\lambda_{X} \circ \Sigma^{*} \langle \id {,} f \diamond f \rangle$
as the act of a context inspecting zero or more subterms performing that two-step
transition. The criterion relates the information obtained by contexts when
inspecting two-step transitions as opposed to inspecting each step individually,
possibly many times over. Specifically, it ensures that the former always
carries \emph{less} information than the latter. A further way to interpret this
criterion is that inspecting an effect \emph{now} instead of \emph{later} does
not produce new outcomes.

\Cref{crit3} is more complex to explain in terms of the simply WB cool format, as
requirements 1,3,4,5 in Definition 8 from van
Glabbeek~\cite{DBLP:journals/tcs/Glabbeek11} all contribute towards observations
on visible transitions not being affected by silent transitions, regardless of
when the latter occur. First, let us look at \emph{straightness}. An operator is
straight if it has no rules where a variable occurs multiple times in the
left-hand side of its premises. The following example shows how non-straight
rules can lead to issues.

\begin{example}[{\cite[Example 3]{DBLP:journals/tcs/Glabbeek11}}]
  \label{ex:glab3}
  Let operator $\textcolor{PineGreen}{\lfloor} \_
  \textcolor{PineGreen}{\rfloor}$ subject to the following rules applying to
  specific $a, b, c \in \Delta$:
  \begin{center}
    $\inference[\textcolor{PineGreen}{pat}]{\goesv{P}{P\pr}{\tau}}{\goesv{\textcolor{PineGreen}{\lfloor} P \textcolor{PineGreen}{\rfloor}}{\textcolor{PineGreen}{\lfloor} P\pr \textcolor{PineGreen}{\rfloor}}{\tau}} \quad
      \inference[\textcolor{PineGreen}{cur}]{\goesv{P}{Q}{a} & \goesv{P}{W}{b}}{\goesv{\textcolor{PineGreen}{\lfloor} P \textcolor{PineGreen}{\rfloor}}{\textcolor{PineGreen}{\lfloor} Q \textcolor{PineGreen}{\rfloor}}{c}}$
  \end{center}
  We can see how rule \textcolor{PineGreen}{cur} violates \cref{crit3} by taking $f(p) = \{(a, q),
  (\tau, w)\}$, $f(q) = \{(\tau, q)\}$ and $f(w) = \{(b, w)\}$. Since
  $(f \diamond f) (p) = \{(a, q),(b, w)\}$, running $\lambda_{X} \circ \Sigma^{*} \langle \id
  {,} f \diamond f \rangle$ on $\textcolor{PineGreen}{\lfloor} p
  \textcolor{PineGreen}{\rfloor}$ induces a $c$-transition, which does not occur
  on the right side.  With rule \textcolor{PineGreen}{cur} weak bisimilarity is not a congruence, as
  $\alpha.\mathtt{0} + b.\mathtt{0}  + \tau.b.\mathtt{0} \approx \alpha.\mathtt{0} +
  \tau.b.\mathtt{0}$
  but $\textcolor{PineGreen}{\lfloor} \alpha.\mathtt{0} + b.\mathtt{0} + \tau.b.\mathtt{0}
  \textcolor{PineGreen}{\rfloor} \not\approx \textcolor{PineGreen}{\lfloor}
  \alpha.\mathtt{0} + \tau.b.\mathtt{0} \textcolor{PineGreen}{\rfloor}$.
\end{example}

Requirements 3 and 4 in the definition of the simply WB cool format underline
how the \emph{lack} of patience rules can affect observations.

\begin{example}[{\cite[Example 5]{DBLP:journals/tcs/Glabbeek11}}]
  \label{ex:glab5}
  Consider operator $\textcolor{Violet}{\lfloor} \_
  \textcolor{Violet}{\rfloor}$ subject to rule
  $\inference[\textcolor{Violet}{oba}]{\goesv{P}{P\pr}{a}}{\goesv{\textcolor{Violet}{\lfloor}
      P \textcolor{Violet}{\rfloor}}{\mathtt{0}}{\tau}}$ applying to a specific
  action $a$. Rule \textcolor{Violet}{oba} fails \cref{crit3} with $f(p) = \{(\tau, q)\}$ and $f(q) = \{(a,
  w)\}$, making $(f \diamond f) (p) = \{(a, w)\}$. Running $\lambda_{X} \circ \Sigma^{*} \langle \id
  {,} f \diamond f \rangle$ on $\textcolor{Violet}{\lfloor} p
  \textcolor{Violet}{\rfloor}$ induces a $\tau$-transition, which does not occur
  on the right side. We can see how  $\alpha.\mathtt{0} \approx
  \tau.\alpha.\mathtt{0}$ but $\textcolor{Violet}{\lfloor} \alpha.\mathtt{0}
  \textcolor{Violet}{\rfloor} \not\approx \textcolor{Violet}{\lfloor}
  \tau.\alpha.\mathtt{0} \textcolor{Violet}{\rfloor}$.
\end{example}

The final requirement is that of \emph{smoothness}. A straight operator for an
LTS is smooth if it has no rules where a variable occurs both in the target and
in the left-hand side of a premise. Non-smooth rules can also cause problems, as
evidenced by the following example.

\begin{example}[{\cite[Example 7]{DBLP:journals/tcs/Glabbeek11}}]
  \label{ex:glab7}
  Let operator $\textcolor{Sepia}{\lfloor} \_ \textcolor{Sepia}{\rfloor}$ with
  the following rules:
  \begin{center}
    $\inference[\textcolor{RawSienna}{play}]{\goesv{P}{P\pr}{\delta}}{\goesv{\textcolor{RawSienna}{\lfloor} P \textcolor{RawSienna}{\rfloor}}{\textcolor{RawSienna}{\lfloor} P\pr \textcolor{RawSienna}{\rfloor}}{\delta}} \quad
    \inference[\textcolor{RawSienna}{pause}]{\goesv{P}{P\pr}{\delta}}{\goesv{\textcolor{RawSienna}{\lfloor} P \textcolor{RawSienna}{\rfloor}}{\textcolor{RawSienna}{\lfloor} P \textcolor{RawSienna}{\rfloor}}{\delta}}$
  \end{center}
  Non-smooth rule \textcolor{RawSienna}{pause} also violates \cref{crit3}. Take
  $f(p) = \{(\tau, q)\}$ and $f(q) = \{(a, w)\}$, making
  $(f \diamond f) (p) = \{(a, w)\}$. Running
  $\lambda_{X} \circ \Sigma^{*} \langle \id {,} f \diamond f \rangle$ on
  $\textcolor{RawSienna}{\lfloor} p \textcolor{RawSienna}{\rfloor}$ induces
  $a$-transition $(a, \textcolor{RawSienna}{\lfloor} p
  \textcolor{RawSienna}{\rfloor})$, but the only $a$-transitions induced on the
  other side are $(a, \textcolor{RawSienna}{\lfloor} q
  \textcolor{RawSienna}{\rfloor})$ and $(a, \textcolor{RawSienna}{\lfloor} w
  \textcolor{RawSienna}{\rfloor})$ instead. Weak bisimilarity fails to be a
  congruence, with $\alpha.\mathtt{0}
  + \tau.b.\mathtt{0} \approx \alpha.\mathtt{0} + \tau.b.\mathtt{0} + b.\mathtt{0}$
  but $\textcolor{PineGreen}{\lfloor} \alpha.\mathtt{0} + \tau.b.\mathtt{0}
  \textcolor{PineGreen}{\rfloor} \not\approx \textcolor{PineGreen}{\lfloor}
  \alpha.\mathtt{0} + \tau.b.\mathtt{0} + b.0 \textcolor{PineGreen}{\rfloor}$.
  The difference here is that only $\textcolor{PineGreen}{\lfloor}
  \alpha.\mathtt{0} + \tau.b.\mathtt{0} + b.0 \textcolor{PineGreen}{\rfloor}$
  is able to perform an $\alpha$-transition after a $b$-transition.
\end{example}

\begin{example}[Continuation of \cref{ex:while}]
  The situation for \cref{crit3} is less obvious for both \texttt{while}-loops
  and sequential composition, but still trivial for \texttt{skip} and assignment
  statements. Using the simpler $\rho$-based formulation of \cref{crit3}, we
  have that for \texttt{while}-loops, the induced transition is not affected
  by transitions of subterms, hence $\rho_{X} \circ \Sigma \langle \id
  {,} f \diamond f \rangle$ is always included in the first iteration of
  $(\lambda_{X} \circ \Sigma^{*} \langle \id {,} f \rangle)^{*} \circ \theta$.

  Showing that the criterion is satisfied by an expression $x \texttt{;} y$ for
  any $x , y \in X$ requires case analysis of $(f
  \diamond f) (x)$ only, as the rule ignores $y$. The two cases are:
  \begin{itemize}
  \item $(t, z) \in (f \diamond f) (x)(s)$. In this case $x$ did not terminate,
    but rather went through an intermediate transition $t\pr, z\pr$. According
    to rule seq2, the transition produced by $\rho_{X} \circ \Sigma
    \langle \id {,} f  \diamond f \rangle$ is $\goes{s, x \texttt{;} y}{{t, z
        \texttt{;} y}}$. Going over
    to $(\lambda_{X} \circ \Sigma^{*} \langle \id {,} f \rangle)^{*} \circ
    \theta$, the first iteration produces $\goes{s, x \texttt{;} y}{{t\pr, z\pr
        \texttt{;} y}}$, and in the second we get $\goes{t\pr, z\pr \texttt{;}
      y}{{t, z \texttt{;} y}}$, which is the same result.
  \item $(t, \checkmark) \in (f \diamond f) (x)(s)$. Here, $x$ terminated
    producing $t$ either immediately $((t, \checkmark) \in f
    (x)(s))$ or in two steps $((s\pr, x\pr) \in f
    (x)(s)$ and $(t, \checkmark) \in f
    (x\pr)(s\pr))$. In any case, the
    transition produced by $\rho_{X} \circ \Sigma \langle \id {,} f
    \diamond f \rangle$ is $\goes{s, x \texttt{;}
      y}{{t, y}}$. Depending on when
    x terminated, this transition will be 
    ``caught'' in either the first or the second iteration of $(\lambda_{X} \circ
    \Sigma^{*} \langle \id {,} f \rangle)^{*} \circ \theta$.
  \end{itemize}
\end{example}

\begin{example}[Continuation of \cref{ex:proc}]
  When instantiated to \spc, the criterion essentially asks if the act of
  ``forgetting'' invisible steps of subterms, as imposed by the rules of monadic
  composition in $\kl{T}$, gives rise to new transitions for a composite term.
  In most cases, this is evidently true; consider for instance the parallel composition of two
  terms $P \Vert Q$, for which
  $\bialg{P}{\tau}{P\pr}{\delta}{P^{\prime\prime}}$, i.e. $(\tau, P\pr) \in
  f(P)$ and $(\delta, P^{\prime\prime}) \in f(P\pr)$.
  Forgetting the invisible step gives $\goesv{P}{P^{\prime\prime}}{\delta}$
  ($(\delta, P^{\prime\prime}) \in (f \diamond f)(P)$), so by rule
  com1 we have $\goesv{P \Vert Q}{P^{\prime\prime} \Vert Q}{\delta}$ on the
  left-hand side. This transition will occur after two iterations on the right-hand
  side, as in $\goesv{P \Vert Q}{\goesv{P\pr \Vert Q}{P^{\prime\prime} \Vert
      Q}{\delta}}{\tau}$.

  Rule syn is especially interesting, as it showcases the full power of
  \cref{crit3}. There are four different cases where syn induces a transition
  $\goesv{P \Vert Q}{P^{\prime\prime} \Vert Q^{\prime\prime}}{\tau}$ on the left side, namely
  \begin{itemize}
  \item $\goesv{P}{\goesv{P\pr}{P^{\prime\prime}}{\alpha}}{\tau}$ and 
    $\goesv{Q}{\goesv{Q\pr}{Q^{\prime\prime}}{\overline{\alpha}}}{\tau}$
  \item $\goesv{P}{\goesv{P\pr}{P^{\prime\prime}}{\alpha}}{\tau}$ and 
    $\goesv{Q}{\goesv{Q\pr}{Q^{\prime\prime}}{\tau}}{\overline{\alpha}}$
  \item $\goesv{P}{\goesv{P\pr}{P^{\prime\prime}}{\tau}}{\alpha}$ and 
    $\goesv{Q}{\goesv{Q\pr}{Q^{\prime\prime}}{\tau}}{\overline{\alpha}}$
  \item $\goesv{P}{\goesv{P\pr}{P^{\prime\prime}}{\tau}}{\alpha}$ and 
    $\goesv{Q}{\goesv{Q\pr}{Q^{\prime\prime}}{\overline{\alpha}}}{\tau}$.
  \end{itemize}
  The third and fourth case work similarly to the first and second
  (resp.) and so we focus on the latter. Either way, the right side of
  \cref{crit3} needs three
  iterations (meaning $(\lambda_{X} \circ \Sigma^{*} \langle \id {,} f \rangle) \diamond
  (\lambda_{X} \circ \Sigma^{*} \langle \id {,} f \rangle) \diamond
  (\rho_{X} \circ \Sigma \langle \id {,} f \rangle)$) in order to produce
  transition $\goesv{P \Vert Q}{P^{\prime\prime} \Vert Q^{\prime\prime}}{\tau}$.
  In the first case, each iteration executes (in sequence) rules com1, com2 and syn leading to
  $\goesv{P \Vert Q}{\goesv{P\pr \Vert Q}{\goesv{P\pr \Vert
        Q\pr}{P^{\prime\prime} \Vert Q^{\prime\prime}}{\tau}}{\tau}}{\tau}$. In
  the second case the order changes with rules com1, syn and com2 inducing
  $\goesv{P \Vert Q}{\goesv{P\pr \Vert Q}{\goesv{P^{\prime\prime} \Vert
        Q\pr}{P^{\prime\prime} \Vert Q^{\prime\prime}}{\tau}}{\tau}}{\tau}$.
\end{example}

Up to this point, the connection of our criteria with the simply WB cool format
has remained informal. The following theorem turns this connection to a formal
correspondence.

\begin{theorem}
  \label{th:simply}
  Any language in the simply WB cool format satisfies the three compositionality
  criteria.
\end{theorem}

The converse is not true, as for example any simple non-smooth rule satisfying
the three criteria, such as $\goesv{[x]}{x}{c}$, is not simply WB cool. A proof
of \cref{th:simply} is provided in~\Cref{subsec:wbproof}.

\subsection{An algebra for a $\dag$}

We are now ready to move on to the main theorem of our work, namely the
existence of a compatible algebra structure for morphism $h^{\dag} : A \to Z$
mapping every term in $A$ to its weak behavior in $Z$. First, we present the
following intermediate result that plays a catalytic role in the main theorem
and is noteworthy in its own right.

\begin{proposition}
  \label{prop:interm}
  Let $\lambda : \Sigma^{*} (\Id \times T)
  \nat T\Sigma^{*}$ be a GSOS law of $\Sigma$ over $T$ with $T$ being an
  $\cpov$-enriched monad. If $\lambda$ satisfies the \emph{three
    compositionality criteria}, then for any $T$-coalgebra $f : X \to TX$ we
  have $(\lambda \circ \Sigma^{*}\langle \id {,} f^{*} \rangle)^{*} = (\lambda
  \circ \Sigma^{*}\langle \id {,} f \rangle)^{*}$.
\end{proposition}
\begin{proof}
      By antisymmetry on the order structure $\leq$ of $T$. To avoid unnecessary
      clutter, for the rest of the proof we shall be writing $\lambda \circ
      \Sigma^{*}\langle \id,f \rangle$ as $\overline{\Sigma}f$ and the $\id$ in
      $\overline{\Sigma}(1 \vee f)$ will stand for $\eta$, the identity
      morphism in $\kl{T}$.
\begin{itemize}
\item Via~\Cref{crit1} and~\Cref{satprop} we have $f \leq f^{*} \implies
  \overline{\Sigma}f \leq
  \overline{\Sigma}f^{*} \implies (\overline{\Sigma}f)^{*} \leq
  (\overline{\Sigma}f^{*})^{*}$. In other words, $(\lambda \circ \Sigma^{*}\langle \id {,} f
  \rangle)^{*} \leq (\lambda \circ \Sigma^{*}\langle \id {,} f^{*} \rangle)^{*}$.
\item We first show that $\overline{\Sigma}(1 \vee f)^{n} \leq
  (\overline{\Sigma}f)^{*}$ for all $n \in \mathbb{N}$ by induction on $n$.
  For $n = 0$ and $n = 1$ and by \Cref{crit1} (as $\id \leq (\id
  \vee f)$) and \Cref{crit2},
  \begin{align}
    \label{ineq1}
    & \overline{\Sigma}(\id \vee f)^{0} = \overline{\Sigma}1 \leq \overline{\Sigma}(1 \vee f) = \overline{\Sigma}(1 \vee f)^{1} \leq (\overline{\Sigma}f)^{*}
  \end{align}
  We now have to show that $\overline{\Sigma}((\id \vee f)^{n+1} \diamond (\id
  \vee f)) \leq (\overline{\Sigma} f)^{*}$ for some $n$ by making use of the
  inductive hypothesis $\overline{\Sigma}(\id \vee f)^{n+1} \leq
  (\overline{\Sigma} f)^{*}$. To that end, we first note that $\id \leq (\id
  \vee f)$ and thus, due to $\cpov$-enrichment, we have that for all $n \in \mathbb{N}$,
  \begin{equation}
    \label{ineqN}
    \id \vee f \leq (\id \vee f) \diamond (\id \vee f) \leq \dots \leq (\id \vee f)^{n+1} \implies (\id \vee f)^{n+1} \diamond (\id \vee f) \leq (\id \vee f)^{n+1} \diamond (\id \vee f)^{n+1}
  \end{equation}
  Next, by \eqref{ineqN}, \Cref{crit1} and \Cref{crit3},
  \begin{equation}
    \label{ineq2}
    \overline{\Sigma}((\id \vee f)^{n+1} \diamond (\id \vee f)) \leq \overline{\Sigma}((\id \vee f)^{n+1} \diamond (\id \vee f)^{n+1}) \leq (\overline{\Sigma}(\id \vee f)^{n+1})^{*}
  \end{equation}
  The induction hypothesis gives $\overline{\Sigma}(\id \vee f)^{n+1}\leq
  (\overline{\Sigma}f)^{*}$ and so \eqref{ineq2} becomes
  \begin{equation}
    \label{proofeq1}
    \overline{\Sigma}((\id \vee f)^{n+1} \diamond (\id \vee f)) \leq (\overline{\Sigma}(\id \vee f)^{n+1})^{*} \leq (\overline{\Sigma}f)^{**} = (\overline{\Sigma}f)^{*}
  \end{equation}
  Inequalities \eqref{proofeq1} and \eqref{ineq1} complete the inductive proof that $\forall
  n \in \mathbb{N}.~\overline{\Sigma}(\id \vee f)^{n} \leq
  (\overline{\Sigma}f)^{*}$. Finally, by~\Cref{sat}
  and~\Cref{crit1}:
  \begin{align*}
    \overline{\Sigma} f^{*} =
    \overline{\Sigma} \bigvee_{n<\omega} (\id \vee f)^{n} =
    \bigvee_{n<\omega} \overline{\Sigma} (\id \vee f)^{n}
  \end{align*}
  We just proved that every link in the $\omega$-chain is smaller than
  $(\overline{\Sigma} f)^{*}$, and thus
  $\overline{\Sigma} f^{*} \leq (\overline{\Sigma} f)^{*}$. By~\Cref{def:free}, this
  becomes $(\overline{\Sigma} f^{*})^{*} \leq (\overline{\Sigma} f)^{*}$, i.e.
  $(\lambda \circ \Sigma^{*}\langle \id {,} f^{*} \rangle)^{*} \leq (\lambda \circ
  \Sigma^{*}\langle \id {,} f \rangle)^{*}$.
\end{itemize}
Having shown both directions, we end up with
$(\lambda \circ \Sigma^{*}\langle \id {,} f^{*}
\rangle)^{*} = (\lambda \circ \Sigma^{*}\langle \id {,} f \rangle)^{*}$.
  \end{proof}

In other words, the transition system of (arbitrarily deep)
contexts with subterms in $X$ is \emph{weakly} equivalent that of contexts
having access to all the weak transitions of
subterms in $X$. \Cref{prop:interm} is what enables the formation of a
compatible algebra structure for $h^{\dag} : A \to Z$ by applying
it to the final coalgebra $z : \goesv{Z}{BZ}{\cong}$. It is currently unclear if
the converse of \Cref{prop:interm} holds.

\begin{theorem}[Main theorem]
  \label{the:main}
  Let $\lambda : \Sigma^{*} (\Id \times T)
  \nat T\Sigma^{*}$ be a GSOS law of $\Sigma$ over $T$ with $T$ being an
  $\cpov$-enriched monad. If $\lambda$ satisfies the \emph{three
    compositionality criteria}, then for $\bialg{\Sigma A}{a}{A}{h}{TA}$ and
  $\bialg{\Sigma Z}{g}{Z}{z}{TZ}$ as in \cref{distr}, $h^{\dag}$ is a
  $\Sigma$-algebra homomorphism from $a : \luarrow{\Sigma A}{A}{\cong}{}$ to
  $z^{\dag} \circ g : \Sigma Z \to Z$.
\end{theorem}
\begin{proof}
  We first observe that the following diagram commutes due to finality of $z$
  (top left and bottom rectangle)
  and naturality of $\lambda$. Note that the horizontal lines are all $T$-coalgebras
  and $g^{\#}$ is the $\mathcal{EM}$-algebra induced by the denotational model
  $g$.
  \begin{center}
    \begin{tikzcd}[row sep=scriptsize]
      \Sigma^{*} Z
      \arrow[r, "\Sigma^{*}\langle \id {,} z^{*} \rangle"]
      \arrow[d, "\Sigma^{*}(z^{\dag})"']
      & \Sigma^{*} (Z \times TZ)
      \arrow[d, "\Sigma^{*}(\Id \times T)(z^{\dag})"]
      \arrow[r, "\lambda"]
      & T\Sigma^{*}Z
      \arrow[d, "T \Sigma^{*} (z^{\dag})"]
      \\
      \Sigma^{*} Z
      \arrow[r, "\Sigma^{*}\langle \id {,} z \rangle"]
      \arrow[d, "g^{\#}"']
      & \Sigma^{*} (Z \times TZ)
      \arrow[r, "\lambda"]
      & T\Sigma^{*}Z
      \arrow[d, "Tg^{\#}"]
      \\
      Z
      \arrow[rr, "z", "\cong"']& &
      TZ
    \end{tikzcd}
\end{center}

  Since rt-closing preserves $T$-coalgebra homomorphisms, rt-closing the
  $T$-coalgebras and finality of $z$ gives us

  \begin{center}
    \begin{tikzcd}[row sep=scriptsize]
      \Sigma^{*} Z
      \arrow[rr, "(\lambda \circ \Sigma^{*}\langle \id {,} z^{*} \rangle)^{*}"]
      \arrow[d, "\Sigma^{*}(z^{\dag})"']
      &
      & T\Sigma^{*}Z
      \arrow[d, "T \Sigma^{*} (z^{\dag})"]
      \\
      \Sigma^{*} Z
      \arrow[rr, "(\lambda \circ \Sigma^{*}\langle \id {,} z \rangle)^{*}"]
      \arrow[d, "g^{\#}"']
      &
      & T\Sigma^{*}Z
      \arrow[d, "Tg^{\#}"]
      \\
      Z
      \arrow[d, "z^{\dag}"']
      \arrow[rr, "z^{*}"]
      & & TZ
      \arrow[d, "T z^{\dag}"]
      \\
      Z
      \arrow[rr, "z", "\cong"']
      & & TZ
    \end{tikzcd}
  \end{center}

  Via \Cref{prop:interm} we have $(\lambda \circ \Sigma^{*}\langle \id {,} z^{*}
  \rangle)^{*} = (\lambda \circ \Sigma^{*}\langle \id {,} z \rangle)^{*}$. Since
  all homomorphisms on final coalgebras are unique, we see that $z^{\dag} \circ
  g^{\#} = z^{\dag} \circ g^{\#} \circ \Sigma^{*}(z^{\dag})$, which in turn
  makes the following diagram commute:
\begin{center}
  \begin{tikzcd}
    \Sigma Z
    \arrow[r, "\theta_{Z}"]
    \arrow[d, "\Sigma z^{\dag}"']
    &
    \Sigma^{*} Z
    \arrow[d, "\Sigma^{*} z^{\dag}"']
    \arrow[r, "g^{\#}"]
    & Z
    \arrow[d, "z^{\dag}"]
    \\
    \Sigma Z
    \arrow[r, "\theta_{Z}"]
    &
    \Sigma^{*} Z
    \arrow[r, "z^{\dag} \circ g^{\#}"]
    & Z
  \end{tikzcd}
\end{center}

Where $\theta : \Sigma \nat \Sigma^{*}$ is the universal natural transformation
sending $\Sigma$ to its free monad. But $g^{\#} \circ \theta = g$ and so we can conclude
\begin{align}
  \label{ff1}
  & z^{\dag} \circ g & = \qquad & z^{\dag} \circ g \circ
                                  \Sigma z^{\dag} & \\
  \label{ff2}
  \implies & z^{\dag} \circ g \circ \Sigma h^{!} & = \qquad & z^{\dag} \circ g \circ
                                                              \Sigma z^{\dag} \circ \Sigma h^{!} & \\
  \label{ff3}
  \implies & z^{\dag} \circ h^{!} \circ a & = \qquad & z^{\dag} \circ g \circ
                                                       \Sigma(z^{\dag} \circ h^{!}) & \\
  \label{ff4}
  \implies & h^{\dag} \circ a & = \qquad & z^{\dag} \circ g \circ \Sigma h^{\dag}.
\end{align}

\cref{ff2} gives \eqref{ff3} due to $h^{!}$ being a bialgebra morphism and finally
we have \eqref{ff4} by~\Cref{lem:weak}, which finishes the proof.
\end{proof}
Weak bisimilarity being a congruence is thus a simple corollary of \Cref{the:main,th:weakcoind}.

\begin{corollary}
  \label{cor:main}
  Let $\lambda : \Sigma^{*} (\Id \times T)
  \nat T\Sigma^{*}$ be a GSOS law of $\Sigma$ over weak-pullback preserving
  functor $T$ with $T$ being an $\cpov$-enriched monad as in \Cref{the:main}. If
  $\lambda$ satisfies the \emph{three compositionality criteria}, weak
  bisimilarity in $\lambda$ is a congruence.
\end{corollary}

\subsection{Lax models for weak bisimulation}
\label{subsec:lax}

Up-to techniques for bisimulation~\cite{DBLP:books/daglib/0067019,
  DBLP:books/cu/12/PousS12} are techniques
that simplify reasoning about behavioral equivalences, the main idea being that
instead of having to show that two processes are included in a bisimulation
relation, one can show that they are included in a different relation which is
\emph{sound} with respect to bisimulation. Bonchi et al.~\cite[Corollary
22]{DBLP:conf/concur/BonchiPPR15} proved that weak bisimulation up-to contextual
closure is a compatible (and hence sound) up-to technique for systems specified
in the simply WB cool rule format. This result is a
corollary of their main theorem~\cite[Theorem 20]{DBLP:conf/concur/BonchiPPR15}, which
requires the underlying system to be \emph{positive} and also the saturated
transition system to be a \emph{lax model} for the given
specification, requirements that are automatically true for systems in the simply WB
cool format.

As mentioned in \cref{subsec:crits}, the abstract equivalent of positivity for GSOS
specifications is monotonicity, which is weaker than continuity
(\Cref{crit1}). Thus, a GSOS law satisfying the three compositionality criteria is
monotone. As for lax models, our behaviors come with an order structure and
hence we can give the following definition.

\begin{definition}[Lax models for GSOS laws]
  \label{def:laxmod}
  Let $\lambda : \Sigma^{*} (\Id \times T)
  \nat T\Sigma^{*}$ be a GSOS law of $\Sigma$ over $T$ with $T$ being an
  $\cpov$-enriched monad. A \emph{lax
  $\lambda$-model} is a $\Sigma$-algebra and $T$-coalgebra pair $\bialg{\Sigma
    X}{g}{X}{h}{TX}$ making the following diagram commute laxly:
  \begin{center}
    \begin{tikzcd}
      \Sigma^{*} X \arrow[r, "g^{\#}"] \arrow[d, "{\Sigma^{*} \langle 1,h \rangle}"']
      \arrow[drr, phantom, "\vertleq"]
      & X \arrow[r, "h"]
      & T X
      \\
      \Sigma^{*} (X \product TX) \arrow[rr, "\lambda"]
      & & T \Sigma^{*} X \arrow[u, "T g^{\#}"']
    \end{tikzcd}
  \end{center}
  Where $g^{\#}$ is the respective $\mathcal{EM}$-algebra induced by $g$.
\end{definition}

In other words, a lax model is a relaxed version of a bialgebra
(\Cref{def:bialg}), implying that \emph{only}, but not necessarily \emph{all}, weak
transitions of a composite term can be deduced from the weak transitions of its
subterms. One non-example of a lax model is $\bialg{\Sigma A}{a}{A}{h^{*}}{BA}$
of~\spc~from \cref{ex:proc}. We can use the same problematic case
as \cref{lax:nonex}: the lower path on the bialgebra diagram for
process $\delta.\texttt{0} + \texttt{0}$ reveals transition
$\goesv{\delta.\texttt{0} + \texttt{0}}{\texttt{0}}{\tau}$, which is not a
transition of $\delta.\texttt{0} + \texttt{0}$.

For all their nice properties, there is little indication as to which GSOS
laws have lax models and why. Thus, in the context of our work, it
is sensible to ask if the three congruence criteria are adequate with respect to
producing lax models for GSOS laws, with the following theorem asserting that
this is indeed the case.

\begin{theorem}
  \label{laxtheorem}
  Let $\lambda : \Sigma^{*} (\Id \times T)
  \nat T\Sigma^{*}$ be a GSOS law of $\Sigma$ over $T$ with $T$ being an
  $\cpov$-enriched monad. If $\lambda$ satisfies the \emph{three compositionality criteria}, then for
  \emph{any} $\lambda$-bialgebra $\bialg{\Sigma X}{g}{X}{h}{TX}$, $\bialg{\Sigma
    X}{g}{X}{h^{*}}{TX}$ is a lax $\lambda$-model.
\end{theorem}
\begin{proof}
  By \cref{satprop} and the definition of a bialgebra we have $T g^{\#} \circ (\lambda_{X}
  \circ \Sigma^{*} \langle \id {,} h \rangle)^{*} = h^{*} \circ g^{\#}$.
  Thus, to prove that $\bialg{\Sigma X}{g}{X}{h^{*}}{TX}$ is lax
  $\lambda$-model, it suffices to show that $\lambda \circ \Sigma^{*}\langle \id {,}
  h^{*} \rangle \leq (\lambda \circ \Sigma^{*}\langle \id {,} h \rangle)^{*}$.
  by \cref{prop:interm} and \Cref{def:free} we indeed have that
  $\lambda \circ \Sigma^{*}\langle \id {,}
  h^{*} \rangle \leq (\lambda \circ \Sigma^{*}\langle \id {,}
  h^{*} \rangle)^{*} = (\lambda \circ \Sigma^{*}\langle \id {,} h
  \rangle)^{*}$.
\end{proof}

It is worth noting that compatibility of the up-to context technique for weak
bisimulation entails the congruence property of weak bisimilarity, which means
that there is an alternative route to~\cref{cor:main} via \cref{prop:interm} and
\cite[Theorem 20]{DBLP:conf/concur/BonchiPPR15}. With that in mind, can lax
models work as a formal method for proving congruence of weak
bisimilarity like our three criteria? The problem is that proving laxness of a
model (typically the initial, rt-closed model $\bialg{\Sigma
  A}{a}{A}{h^{*}}{TA}$) involves non-trivial reasoning on an
rt-closed system that is itself defined inductively on the structure
of terms. Conversely, our three criteria are significantly easier to
establish as they characterize a GSOS law acting on a single layer of syntax.


\section{Conclusion}
In this paper we presented three abstract criteria over operational
semantics, given in the form of Turi and Plotkin's bialgebraic semantics, that
guarantee weak bisimilarity being a congruence. We believe that the criteria
gracefully balance between generality and usefulness but, as is often the case
with abstract results, this is something hard to assess accurately. What is
equally important however is that each of the criteria can be given an
intuitive explanation as to the kind of restriction it imposes on the
semantics. We hope that these insights can contribute towards a conclusive
answer to the general problem of full abstraction: the definition of the best
adequate denotational semantics, the underlying equivalence of which coincides
with contextual equivalence.




\bibliography{mainBiblio}

\appendix

\section{The simply WB cool rule format}
\label{app:simply}
In this section we introduce a few notions relevant to the simply WB cool rule
format taken from~\cite{DBLP:journals/tcs/Glabbeek11}. For some $n$-ary operator $o$
  in a language $\mathcal{L}$, the \emph{rules of o} are all the rules with
  source $o(x_{1},\dots,x_{n})$. The following characterize GSOS rules based on
  the shape(s) of its premise(s) and the shape of its conclusion.

  \begin{itemize}
  \item An operator in $\mathcal{L}$ is \emph{straight} if it has no rules in which a
    variable occurs multiple times in the left-hand side of its premises.
  \item An operator is \emph{smooth} if it is straight and has no rules in which
    a variable occurs both in the target and in the left-hand side of a premise.
  \item An argument $i \in \mathbb{N}$ of an operator $o$ is \emph{active} if
    there is a rule of $o$ in which $x_{i}$ is the left-hand side of a premise.
  \item A variable $x$ occurring in a term $t$ is \emph{receiving in $t$} if $t$ is the
    target of a rule in $\mathcal{L}$ in which $x$ is the right-hand side of
    a premise. An argument $i \in \mathbb{N}$ of an operator $o$ is receiving if
    a variable $x$ is receiving in a term $t$ that has a subterm
    $o(v_{1},\dots,v_{n})$ with $x$ occurring in $v_{i}$.
  \item A rule of the form of
    $\inference{\goesv{x_{i}}{y}{\tau}}{\goesv{o(x_{1},\dots,x_{n})}{o(x_{1},\dots,x_{n})[y/x_{i}]}{\tau}}$
    for $1 \leq i \leq n$ is called a patience rule for the $i$th argument of
    $o$, where $t[y/x]$ stands for term $t$ with all occurrences of $x$ replaced
    by $y$.
  \end{itemize}

  We can now give the complete definition of the simply WB cool rule format.

  \begin{definition}
    \label{def:simply}
    A GSOS language $\mathcal{L}$ is simply WB cool if it is positive and the
    following conditions are all true.
    \begin{enumerate}
    \item All operators in $\mathcal{L}$ are straight.
    \item The only rules in $\mathcal{L}$ with $\tau$-premises are patience rules.
    \item Every active argument of an operator has a patience rule.
    \item Every receiving argument of an operator has a patience rule.
    \item All operators in $\mathcal{L}$ are smooth.
    \end{enumerate}
  \end{definition}

\section{Selected proofs}
\label{sec:proofs}

In this section of the appendix we include the proofs of \cref{prop:cont} and
\cref{th:simply}.

\subsection{Equivalence of the two representations of the continuity criterion}

\begin{proof}[Proof of \cref{prop:cont}]
  We introduce the friendlier notation $\overline{\Sigma}f$ in place of
  $\lambda \circ \Sigma^{*} \langle \id, f\rangle$ for any $f : X \to TX$.
  \eqref{eq:crit1} $\implies$ \eqref{eq:crit1a} is immediate since $\rho = \lambda \circ \theta$, where
  $\theta : \Sigma \nat \Sigma^{*}$ is the universal natural transformation
  sending $\Sigma$ to its free monad. For \eqref{eq:crit1a} $\implies$
  \eqref{eq:crit1}, we first note that for each ``link'' $f_{i}$,
  $\overline{\Sigma}f_{i}$ is (equivalently) defined via ``structural recursion with
  accumulators'' (see~\cite[Theorem 5.1]{DBLP:conf/lics/TuriP97}), i.e. it is
  the unique morphism making the following diagram commute.
  \begin{center}
    \begin{tikzcd}
      \Sigma \Sigma^{*} X \arrow[rr, "\mu \circ \theta_{\Sigma^{*}}"] \arrow[d, "{\Sigma
        \langle 1,\overline{\Sigma}f_{i} \rangle}"]
      && \Sigma^{*} X \arrow[d, "{\overline{\Sigma}f_{i}}"',dashed, "\exists !"]
      & X \arrow[l, "\eta"'] \arrow[dl, "T\eta \circ f_{i}"]
      \\
      \Sigma (\Id \product T)\Sigma^{*} X
      \arrow[r, "\rho_{\Sigma^{*}}"]
     & T \Sigma^{*}\Sigma^{*}X \arrow[r, "T\mu"]
     & T \Sigma^{*} X
    \end{tikzcd}
	\end{center}

  This makes $\overline{\Sigma}f_{i}$ the (necessarily unique) \emph{homomorphic
    extension} of $T\mu \circ \rho_{\Sigma^{*}}$ along
  $T\eta \circ f_{i}$. Continuity of composition in $\kl{T}$ gives
  \begin{equation}
    \label{eq:crA}
    T\eta \circ \bigvee_{i} f_{i} = \bigvee_{i} (T\eta \circ f_{i})
  \end{equation}

  $T\eta \circ f_{i} = \overline{\Sigma}f_{i} \circ \eta$, thus $\bigvee_{i}
  (\overline{\Sigma}f_{i} \circ \eta)$ exists and also
  \begin{equation}
    \label{eq:crB}
    \bigvee_{i} (\overline{\Sigma}f_{i} \circ \eta) = (\bigvee_{i} \overline{\Sigma}f_{i}) \circ \eta
  \end{equation}

  Via similar reasoning, we have
  \begin{gather}
    \label{prcrit:eq2}
    T\mu \circ \bigvee_{i}
    (\rho_{\Sigma^{*}} \circ \Sigma \langle \id{,}\overline{\Sigma}f_{i}
    \rangle) = \bigvee_{i} (T\mu \circ \rho_{\Sigma^{*}} \circ \Sigma \langle
    \id{,}\overline{\Sigma}f_{i} \rangle) \\
    \label{prcrit:eq1}
    \bigvee_{i} (\overline{\Sigma}f_{i} \circ \mu \circ \theta_{\Sigma^{*}}) = (\bigvee_{i}
    \overline{\Sigma}f_{i}) \circ \mu \circ \theta_{\Sigma^{*}}
  \end{gather}

  \cref{eq:crit1a} allows us to rewrite \cref{prcrit:eq2} as
  \begin{equation}
    \label{prcrit:eq3}
    T\mu \circ \rho_{\Sigma^{*}} \circ \Sigma \langle \id{,} \bigvee_{i} \overline{\Sigma}f_{i} \rangle = \bigvee_{i} (T\mu \circ \rho_{\Sigma^{*}} \circ \Sigma \langle \id{,}\overline{\Sigma}f_{i} \rangle)
  \end{equation}

  By taking the supremum over $i$ of the $i$-dependent arrows in the previous diagram,
  \cref{eq:crA}, \eqref{eq:crB}, \eqref{prcrit:eq1} and \eqref{prcrit:eq3} allow
  us to present $\bigvee_{i} \overline{\Sigma}f_{i}$ as a morphism that makes
  the following diagram commute.
  \begin{center}
    \begin{tikzcd}
      \Sigma \Sigma^{*} X \arrow[rr, "\mu \circ \theta_{\Sigma^{*}}"] \arrow[d, "{\Sigma
        \langle 1,\bigvee_{i} \overline{\Sigma}f_{i} \rangle}"]
      && \Sigma^{*} X \arrow[d, "{\bigvee_{i} \overline{\Sigma}f_{i}}"']
      & X \arrow[l, "\eta"'] \arrow[dl, "T\eta \circ \bigvee_{i} f_{i}"]
      \\
      \Sigma (\Id \product T)\Sigma^{*} X \arrow[r, "\rho_{\Sigma^{*}}"]
      & T \Sigma^{*}\Sigma^{*}X \arrow[r, "T\mu"]
      & T \Sigma^{*} X
    \end{tikzcd}
	\end{center}

  But there can only be one such morphism, namely $\overline{\Sigma}(\bigvee_{i}
  f_{i})$, the unique homomorphic extension of $T\mu \circ
  \rho_{\Sigma^{*}}$ along $T\eta \circ \bigvee_{i} f_{i}$. In other words,
  $\bigvee_{i} \overline{\Sigma}f_{i} = \overline{\Sigma}(\bigvee_{i} f_{i})$.
\end{proof}

\subsection{Correspondence with the simply WB cool rule format}
\label{subsec:wbproof}

\begin{proof}[Proof of \cref{th:simply}]
  In order to prove that any language $\mathcal{L}$ in the simply WB cool format
  automatically satisfies the three criteria, it suffices to show that (the
  GSOS law induced by) any arbitrary rule in $\mathcal{L}$ satisfies them. First
  of all, we know that rules in the simply WB cool format are of the form
  $\inference{\{\goesv{x_{i}}{y_{i}}{c_{i}}~|~i \in
    I\}}{\goesv{o(x_{1},\dots,x_{n})}{t}{\alpha}}$ for $I \subseteq
  \{1,\dots,n\}$\cite[\S 3]{DBLP:journals/tcs/Glabbeek11}, or simply
  $\inference{\{\goesv{x_{i}}{y_{i}}{c_{i}}~|~i \in
    I\}}{\goesv{o(\overrightarrow{x})}{t}{\alpha}}$, where $o$ is an
  $n$-ary operator. We thus consider an arbitrary rule in the above
  form and proceed by distinguishing by the number of active arguments.

  \subsubsection{0 active arguments}
  All cases are trivial.

  \subsubsection{1 active argument}
  The rule is of the form
  $\inference{\goesv{x_{j}}{y}{c}}{\goesv{o(\overrightarrow{x})}{t}{\alpha}}$
  meaning that the rule is active
  on a position $j$. There is only a single premise because of the first
  requirement (straightness) in \cref{def:simply}.

  \paragraph{Patience rule}
  \label{par:pat}

  If $c = \tau$ then by the second requirement of \cref{def:simply} this rule
  has to be a patience rule of the form
  $\inference{\goesv{x_{j}}{y}{\tau}}{\goesv{o(\overrightarrow{x})}{o(\overrightarrow{x})[y/x_{j}]}{\tau}}$.

  \subparagraph{\cref{crit1}}
  \begin{gather*}
    \bigvee_{i} \rho (o(\overrightarrow{x}, f_{i}(\overrightarrow{x}))) =
    \bigvee_{i} \{ \tau, o(\overrightarrow{x})[y/x_{j}]~|~(\tau, y_{i}) \in
    f_{i}(x_{j})\} =  \\ \{ \tau, o(\overrightarrow{x})[y/x_{j}]~|~(\tau, y_{i}) \in
    \bigvee_{i} (f _{i}(x_{j}))\} = \rho (o(\overrightarrow{x}, \bigvee_{i} (f_{i}(\overrightarrow{x})))
  \end{gather*}

  \subparagraph{\cref{crit2}}
  $\rho (o(\overrightarrow{x}, (\eta \vee f)(\overrightarrow{x})))$ induces a
  single transition
  $\goesv{o(\overrightarrow{x})}{o(\overrightarrow{x})}{\tau}$ for all $f$,
  which is included in $\eta_{x} \circ \theta_{X}$.

  \subparagraph{\cref{crit3}}

  The only transitions in $\rho (o(\overrightarrow{x}, (f \diamond
  f)(\overrightarrow{x})))$ are
  $\goesv{o(\overrightarrow{x})}{o(\overrightarrow{x})[z/x_{j}]}{\tau}$ when
  $\goesv{x_{j}}{\goesv{y}{z}{\tau}}{\tau}$, i.e. $(\tau, y) \in
  f(x_{j}), (\tau, z) \in g(y)$ for some $y,z$. Running $(\lambda_{X} \circ
  \Sigma^{*} \langle \id {,} f \rangle) \diamond
  (\rho_{X} \circ \Sigma \langle \id {,} f
  \rangle)$~\footnote{Recall that $(\lambda_{X} \circ \Sigma^{*} \langle \id {,} f
    \rangle) \circ \theta = \rho_{X} \circ \Sigma \langle \id {,} f
    \rangle$.},
  which is the second iteration on the right side, we can see that
  $\goesv{o(\overrightarrow{x})}{\goesv{o(\overrightarrow{x})[y/x_{j}]}{o(\overrightarrow{x})[z/y]}{\tau}}{\tau}$,
  which satisfies the criterion.

  \paragraph{Impatient rule}
  \label{par:impat}

  This time we have $c \neq \tau$ which entails, by the third requirement of
  \cref{def:simply}, the presence of a patience rule for argument $j$.

  \subparagraph{\cref{crit1}}
  Similar to \ref{par:pat}.
  \begin{gather*}
    \bigvee_{i} \rho (o(\overrightarrow{x}, f_{i}(\overrightarrow{x}))) =
    \bigvee_{i} \{c, o(\overrightarrow{x})[y/x_{j}]~|~(c, y_{i}) \in
    f_{i}(x_{j})\} =  \\ \{c, o(\overrightarrow{x})[y/x_{j}]~|~(c, y_{i}) \in
    \bigvee_{i} (f _{i}(x_{j}))\} = \rho (o(\overrightarrow{x}, \bigvee_{i} (f_{i}(\overrightarrow{x})))
  \end{gather*}

  \subparagraph{\cref{crit2}}
  Similarly to \ref{par:pat}, the presence of the patience rule for argument $j$
  means that there is always transition
  $\goesv{o(\overrightarrow{x})}{o(\overrightarrow{x})}{\tau}$ for all $f$ on
  the left side, which is included on the right side by $\eta_{x} \circ
  \theta_{X}$. The transitions induced by $f$ exist on both sides.

  \subparagraph{\cref{crit3}}
  Transitions on the left side occur if and only if there are $w, z$ such that
  $\goesv{x_{j}}{\goesv{w}{z}{c}}{\tau}$ or
  $\goesv{x_{j}}{\goesv{w}{z}{\tau}}{c}$, i.e. $(\tau, w) \in
  f(x_{j}), (c, z) \in f(w)$ or $(c, w) \in
  f(x_{j}), (\tau, z) \in f(w)$. We also have to distinguish between $y$ in premise
 $\goesv{x_{j}}{y}{c}$ being receiving in $t$ or not. For each case, we give the
 transition(s) on the left side (of \cref{crit3}) and the respective iteration step where the
 left-side transitions appear on the right side (of \cref{crit3}).
 \begin{enumerate}
 \item $y$ is not receiving, $\goesv{x_{j}}{\goesv{w}{z}{c}}{\tau}$.
   \begin{itemize}
   \item $\rho_{X} \circ \Sigma \langle \id
     {,} f \diamond f \rangle$: $\goesv{o(\overrightarrow{x})}{t}{\alpha}$
   \item $(\lambda_{X} \circ \Sigma^{*} \langle \id {,} f \rangle) \diamond
     (\rho_{X} \circ \Sigma \langle \id {,} f \rangle)$:
     $\goesv{o(\overrightarrow{x})}{\goesv{o(\overrightarrow{x})[w/x_{j}]}{t}{\alpha}}{\tau}$
   \end{itemize}
 \item $y$ is not receiving, $\goesv{x_{j}}{\goesv{w}{z}{\tau}}{c}$.
   \begin{itemize}
     \item $\rho_{X} \circ \Sigma \langle \id
       {,} f \diamond f \rangle$: $\goesv{o(\overrightarrow{x})}{t}{\alpha}$
     \item $\rho_{X} \circ \Sigma \langle \id {,} f \rangle$:
       $\goesv{o(\overrightarrow{x})}{t}{\alpha}$
 \end{itemize}
\item $y$ is receiving, $\goesv{x_{j}}{\goesv{w}{z}{c}}{\tau}$.
  \begin{itemize}
  \item $\rho_{X} \circ \Sigma \langle \id
    {,} f \diamond f \rangle$: $\goesv{o(\overrightarrow{x})}{t(z)}{\alpha}$
  \item $(\lambda_{X} \circ \Sigma^{*} \langle \id {,} f \rangle) \diamond
    (\rho_{X} \circ \Sigma \langle \id {,} f \rangle)$:
    $\goesv{o(\overrightarrow{x})}{\goesv{o(\overrightarrow{x})[w/x_{j}]}{t(z)}{\alpha}}{\tau}$
   \end{itemize}
 \item $y$ is receiving, $\goesv{x_{j}}{\goesv{w}{z}{\tau}}{c}$ \\
   This is the trickiest case. Let us look back at the rule in question (with
   $y$ receiving), which is
   $\inference{\goesv{x_{j}}{y}{c}}{\goesv{o(\overrightarrow{x})}{t(y)}{\alpha}}$.
   The key observation is that requirement 4 in \cref{def:simply}, requiring that every
   receiving argument of an operator has a patience rule, implies that no matter
   how complex the receiving expression $t(y)$ is, there will be patience rules in
   place to ensure a derivation amounting to
   $\inference{\goesv{x}{y}{\tau}}{\goesv{s(x)}{s(y)}{\tau}}$ for each
   sub-expression $s(y)$ in $t(y)$ that ``receives'' $y$. The number of
   iterations on the right-hand required in order to trigger all
   necessary patience rules depends on the number of sub-expressions $s$.

   For example, let $t(y) \triangleq d(e(l(y,0)),e(l(0,y)))$, where $d,l$ are
   binary operations, $e$ is a unary operation and $0$ is some term. Due to
   requirement 4, $d, l, e$ will all have patience rules in all positions and we
   need exactly two iterations to trigger the two patience rules in each of the
   positions of $d$. More generally,
   \begin{itemize}
   \item $\rho_{X} \circ \Sigma \langle \id
     {,} f \diamond f \rangle$: $\goesv{o(\overrightarrow{x})}{t(z)}{\alpha}$
   \item $(\lambda_{X} \circ \Sigma^{*} \langle \id {,} f \rangle)^{*} \diamond
     (\rho_{X} \circ \Sigma \langle \id {,} f \rangle)$:
     $\goesv{o(\overrightarrow{x})}{\goesv{t(w)}{t(z)}{\tau^{*}}}{\alpha}$.
   \end{itemize}
 \end{enumerate}
 
 \subsubsection{2 or more active arguments}
 
 We first note that the only rules with $\tau$-premises are
 patience rules, which are already covered in \ref{par:pat} and so we move on
 to $c_{1},c_{2} \not = \tau$ with the rule being of the form of
 $\inference{\goesv{x_{i}}{y_{1}}{c_{1}} &
   \goesv{x_{j}}{y_{2}}{c_{2}}}{\goesv{o(\overrightarrow{x})}{t}{\alpha}}$. The
 third requirement of \cref{def:simply} means that there are patience rules for
 arguments $i$ and $j$ in $o$.

 \subparagraph{\cref{crit1}}
 Similar to the case for \cref{crit1} in \ref{par:pat}.

 \subparagraph{\cref{crit2}}
 Similar to the case for \cref{crit2} in \ref{par:pat}.

 \subparagraph{\cref{crit3}}
 There are four separate cases where a transition occurs
 in $\rho (o(\overrightarrow{x}, (f \diamond f)(\overrightarrow{x})))$, as a
 $c_{1}$-transition and a $c_{2}$-transition may occur in either step for both subterms.
 \begin{enumerate}
 \item $\goesv{x_{i}}{\goesv{w_{1}}{z_{1}}{c_{1}}}{\tau}$ and 
   $\goesv{x_{j}}{\goesv{w_{2}}{z_{2}}{c_{2}}}{\tau}$
 \item $\goesv{x_{i}}{\goesv{w_{1}}{z_{1}}{c_{1}}}{\tau}$ and 
   $\goesv{x_{j}}{\goesv{w_{2}}{z_{2}}{\tau}}{c_{2}}$
 \item $\goesv{x_{i}}{\goesv{w_{1}}{z_{1}}{\tau}}{c_{1}}$ and 
   $\goesv{x_{j}}{\goesv{w_{2}}{z_{2}}{\tau}}{c_{2}}$
 \item $\goesv{x_{i}}{\goesv{w_{1}}{z_{1}}{\tau}}{c_{1}}$ and 
   $\goesv{x_{j}}{\goesv{w_{2}}{z_{2}}{c_{2}}}{\tau}$.
 \end{enumerate}
 In addition, $y_{1}$ and $y_{2}$ can each be either receiving or not receiving
 in $t$ and, as this affects transitions the same way as $y$ being receiving in
 \ref{par:impat}.
 \begin{enumerate}
 \item $\goesv{x_{i}}{\goesv{w_{1}}{z_{1}}{c_{1}}}{\tau}$ and 
   $\goesv{x_{j}}{\goesv{w_{2}}{z_{2}}{c_{2}}}{\tau}$

   Here, whether $y_{1}$ and $y_{2}$ are receiving or not does not make a
   difference and we write $t(y_{1},y_{2})$ to denote a term which potentially has instances
   of $y_{1}$ and $y_{2}$.
   \begin{itemize}
   \item $\rho_{X} \circ \Sigma \langle \id
     {,} f \diamond f \rangle$: $\goesv{o(\overrightarrow{x})}{t(z_{1},z_{2})}{\alpha}$
   \item
     $(\lambda_{X} \circ \Sigma^{*} \langle \id {,} f \rangle) \diamond
     (\lambda_{X} \circ \Sigma^{*} \langle \id {,} f \rangle) \diamond
     (\rho_{X} \circ \Sigma \langle \id {,} f \rangle)$: \\
     $\goesv{o(\overrightarrow{x})}{\goesv{o(\overrightarrow{x})[w_{1}/x_{i}]}{\goesv{o(\overrightarrow{x})[w_{1}/x_{i}][w_{2}/x_{j}]}{t(z_{1},z_{2})}{\alpha}}{\tau}}{\tau}$
   \end{itemize}

   We can see that we need to iterate three times on the right-hand side: one to
   trigger the patience rule on position $i$, one to trigger the patience rule
   on position $j$ on the new term and one more to trigger the main rule.
   
 \item $\goesv{x_{i}}{\goesv{w_{1}}{z_{1}}{c_{1}}}{\tau}$ and 
   $\goesv{x_{j}}{\goesv{w_{2}}{z_{2}}{\tau}}{c_{2}}$

   In this case $y_{2}$ being receiving makes a difference, while $y_{1}$ does not. Let us
   first deal with the non-receiving case for $y_{2}$.
   \begin{itemize}
   \item $\rho_{X} \circ \Sigma \langle \id
     {,} f \diamond f \rangle$: $\goesv{o(\overrightarrow{x})}{t(z_{1})}{\alpha}$
   \item
     $ (\lambda_{X} \circ \Sigma^{*} \langle \id {,} f \rangle) \diamond
     (\rho_{X} \circ \Sigma \langle \id {,} f \rangle)$: $\goesv{o(\overrightarrow{x})}{\goesv{o(\overrightarrow{x})[w_{1}/x_{i}]}{t(z_{1})}{\alpha}}{\tau}$
   \end{itemize}
   The first iteration will trigger the patience rule for $i$, while the second
   produces the $\alpha$-transition. Variable $y_{2}$ is not receiving and so
   nothing else is needed. On the other hand, if $y_{2}$ is receiving, we see that
   \begin{itemize}
   \item $\rho_{X} \circ \Sigma \langle \id
     {,} f \diamond f \rangle$: $\goesv{o(\overrightarrow{x})}{t(z_{1},z_{2})}{\alpha}$
   \item
     $(\lambda_{X} \circ \Sigma^{*} \langle \id {,} f \rangle)^{*} \diamond
     (\lambda_{X} \circ \Sigma^{*} \langle \id {,} f \rangle) \diamond
     (\rho_{X} \circ \Sigma \langle \id {,} f \rangle)$: \\
     $\goesv{o(\overrightarrow{x})}{\goesv{o(\overrightarrow{x})[w_{1}/x_{i}]}{\goesv{t(z_{1},
           w_{2})}{t(z_{1},z_{2})}{\tau^{*}}}{\alpha}}{\tau}$
   \end{itemize}
   Similarly to the fourth case of \cref{crit3} in \ref{par:impat}, requirement
   4 in \cref{def:simply} guarantees that there will be a sequence of patience
   rules that gives $\goesv{t(z_{1}, w_{2})}{t(z_{1},z_{2})}{\tau^{*}}$.
 \item $\goesv{x_{i}}{\goesv{w_{1}}{z_{1}}{\tau}}{c_{1}}$ and 
   $\goesv{x_{j}}{\goesv{w_{2}}{z_{2}}{\tau}}{c_{2}}$

   If $y_{1}$ and $y_{2}$ are not receiving, this is very similar to the first
   case. If $y_{1}$ and/or $y_{2}$ are receiving, then requirement 4 in
   \cref{def:simply} comes into play in the same manner as before. For instance,
   if both $y_{1}$ and $y_{2}$ are receiving:
   \begin{itemize}
   \item $\rho_{X} \circ \Sigma \langle \id
     {,} f \diamond f \rangle$: $\goesv{o(\overrightarrow{x})}{t(z_{1},z_{2})}{\alpha}$
   \item
     $(\lambda_{X} \circ \Sigma^{*} \langle \id {,} f \rangle)^{*} \diamond
     (\lambda_{X} \circ \Sigma^{*} \langle \id {,} f \rangle)^{*} \diamond
     (\rho_{X} \circ \Sigma \langle \id {,} f \rangle)$: \\
     $\goesv{o(\overrightarrow{x})}{\goesv{t(w_{1},w_{2})}{\goesv{t(z_{1},
           w_{2})}{t(z_{1},z_{2})}{\tau^{*}}}{\tau^{*}}}{\alpha}$
   \end{itemize}
 \item $\goesv{x_{i}}{\goesv{w_{1}}{z_{1}}{\tau}}{c_{1}}$ and 
   $\goesv{x_{j}}{\goesv{w_{2}}{z_{2}}{c_{2}}}{\tau}$
   
   Very similar to the second case.
 \end{enumerate}

 In the presence of three or more active arguments we can apply the same principles,
 the only difference being that the maximum number of necessary iterations on
 the right side will be higher, as more patience rules will have to be triggered.
\end{proof}


\end{document}